\documentclass[a4paper]{llncs}
\usepackage{lmodern}
\usepackage{fullpage}

\usepackage{outlines}
\usepackage{amsmath}
\usepackage{amsfonts}
\usepackage{mathtools}
\usepackage{cite}
\usepackage[colorlinks = false]{hyperref}
\usepackage{todonotes}

\usepackage{enumitem}
\setlist[enumerate,1]{label=\alph*)}

\newcommand{\eps}{\varepsilon}
\newcommand{\sig}{\sigma}

\newcommand{\oeps}{1+\eps}
\newcommand{\ieps}{1/\eps}
\newcommand{\opt}{\mathsf{OPT}}
\newcommand{\cb}{\textsf{K}}
\newcommand{\id}{\textsf{U}}

\newcommand{\sma}{\mathsf{S}}
\newcommand{\lar}{\mathsf{L}}
\newcommand{\larsiz}{\sizes}
\newcommand{\jai}{\textsf{SUP}}
\newcommand{\bpu}{\textsf{BUP}}
\newcommand{\jaii}{\textsf{RP}}
\newcommand{\tp}{p^{\dagger}}
\newcommand{\ms}{s^{\dagger}}
\newcommand{\sizes}{\mathsf{Z}}
\newcommand{\cons}{\mathsf{CN}}

\newcommand{\jobs}{\textsf{J}}
\newcommand{\machines}{\textsf{M}}
\newcommand{\items}{\mathsf{T}}
\newcommand{\early}{\psi}
\newcommand{\late}{\phi}

\newcommand{\pre}{\mathbb{Q}_{>0}}
\newcommand{\nre}{\mathbb{Q}_{\geq0}}
\newcommand{\ints}{\mathbb{Z}}
\newcommand{\pint}{\mathbb{Z}_{>0}}
\newcommand{\nint}{\mathbb{Z}_{\geq0}}

\title{Scheduling with cardinality dependent unavailability periods\thanks{Supported in part by ISF - Israel Science Foundation grant numbers 308/18 and 1467/22.}}
\author{G. Jaykrishnan\inst{1} \and Asaf Levin\inst{1}}

\authorrunning{G. Jaykrishnan \and A. Levin}

\institute{Faculty of Data and Decision Sciences, The Technion, Haifa, Israel.\email{jaykrishnang@hotmail.com.} \and
	Faculty of Data and Decision Sciences, The Technion, Haifa, Israel.\email{levinas@technion.ac.il.}}
\date{ }

\begin{document}
\maketitle
\begin{abstract}
We consider non-preemptive scheduling problems on parallel identical machines where machines change their status from being available to being unavailable and vice versa along the time horizon. The particular form of unavailability we consider is when the starting time of each downtime depends upon the cardinality of the job subset processed on that machine since the previous downtime. We consider the problem of  minimizing the makespan in such scenarios as well as its dual problem where we have a fixed common deadline of $1$ and the goal is to minimize the number of machines for which there is a feasible schedule. We develop an EPTAS for the first variant and an AFPTAS for the second variant.
\end{abstract}

\section{Introduction}
Scheduling of independent jobs on parallel machines with the goal of minimizing the makespan as well as bin packing are well-studied problems. Scheduling involves assigning jobs to machines, whereas bin packing involves assigning jobs to machines (or items to bins) subject to a common deadline of $1$ and the goal of minimizing the number of machines we use. In real life scenarios the machines may not be continuously available to process the jobs. For instance, consider production systems where maintenance, modification or repair of machines are needed.  Our aim with this work is to consider the scheduling variant and the bin packing variant of the problem which involves unavailability periods whose starting times depend on the number of jobs scheduled on the same machine since its previous unavailability period.  We consider only non-preemptive scheduling. We start with notation, problem definitions, and related results.

\subsubsection{Notation.}
For this discussion all the numbers we deal with are rational numbers. We denote the set of non-negative rationals, positive rationals, non-negative integers, and positive integers with $\nre$, $\pre$, $\nint$, and $\pint$, respectively. Let $\eps\in\pre$ be such that $\ieps\in\ints$. $[n]$ for $n\in\pint$ is the set of natural numbers up to $n$, namely, $[n] = \{1,\ldots,n\}$. The input instance is denoted by $I$ for both variants.  

\subsubsection{Problem definitions.}  In this work we consider two variants of the model, the first one is the scheduling variant and the second is the bin packing variant.
First, we define the scheduling variant. The input instance consists of a set $\jobs$ of $n$ jobs where each job $j$ has a size $p_j\in\nre$, a set $\machines$ of $m$ parallel identical machines, a positive integer $k$, and the duration of each unavailability period $\id \in\nre$. The processing time of a job on a machine is equal to its size. In this problem, every machine after processing $k$ consecutive jobs becomes unavailable for a time period of length $\id$. After the idle time it becomes available to start processing up to $k$ jobs again and so on.   This condition, for ease, is represented as the \emph{unavailability requirement}. There is no idle time after the last at most $k$ consecutive jobs on the machine. The \emph{load} of a machine is defined as the time at which the last job assigned to the machine finishes processing.   That is, the load of a machine as a result of the unavailability requirement is $$ \sum_{j\in\sig}p_j + \id\left\lfloor \frac{|\sig|-1}{k}\right\rfloor,$$ where $\sig$ is the set of jobs assigned to the machine. These $|\sig|$ jobs can be assigned in any permutation and the load of the machine does not change by the selection of the permutation.  This is so since the load only depends on the number of jobs and the total processing time of these jobs. Hence, we assume that the jobs are scheduled in a non-increasing order of size on each machine. A \emph{feasible solution} to this problem is a partition of $\jobs$ into $m$ subsets. The \emph{makespan} of a solution is the largest load among all the partitions. The objective function is to minimize the makespan.

This problem introduced in \cite{liao2007minimizing} is called \textsc{Scheduling with Unavailability Periods} abbreviated as \jai{}.  Note that if $k\geq n$, then \jai\ is exactly the makespan minimization problem on identical machines that is known to be strongly NP-hard so this applies also for \jai.  Furthermore, if $U>> \sum_{j\in \jobs} p_j$ and $n\leq mk$, then it suffices to consider solutions for \jai\ satisfying  that each machine is assigned at most $k$ jobs, and we get the cardinality constrained scheduling problem that was considered in the literature (see below).

Next, we define the bin packing variant. The input instance contains a set of $n$ items $\items$ where each item $i$ has a size $s_i\in[0,1]$, infinite number of available bins of size $1$, a rational number $\id\in(0,1]$, and a positive integer $k$. After packing $k$ consecutive items in a common bin, an \emph{idle item} (of size $\id$) is added to the bin. This is not carried out for the last at most $k$ consecutive items. The packing of a bin is feasible if the total size of all items in the bin, including the idle items, is at most $1$. 
This requirement is represented as the \emph{unavailability requirement}. As a consequence of the idle items, we define the \emph{load} of a set of items $p$ (packed into a single bin) as $$\sum_{i\in p}s_i +\left\lfloor\frac{|p|-1}{k}\right\rfloor\id.$$ Since the load only depends on the total size of the items in the bin and their number, it is invariant to the permutation of the items. Thus, we assume that the items in each bin are sorted in a non-increasing order of size. A \emph{feasible solution} is a partition of the set $\items$ into bins such that the load of each bin is at most $1$. The objective is to minimize the number of bins in a feasible solution. This problem is called \textsc{Bin packing with Unavailability Periods} abbreviated as \bpu{}.  Similarly to \jai, observe the following relations to the well-studied special cases of \bpu.  First, if $k\geq n$, then there are no idle items, so the resulting special case of \bpu\ is the standard bin packing problem.  Furthermore, if $U=1$, then every bin is assigned at most $k$ items, so the resulting special case of \bpu\ is cardinality constrained bin packing.

Observe that, as usual, checking if a makespan of $1$ is achievable in \jai{} is the same as checking if the optimal solution for  \bpu{} has at most $m$ bins; thus, for optimal solutions the two problems are equivalent.

\subsubsection{Definitions of approximation algorithms both in the absolute case and in the asymptotic case.} 
We define the \emph{asymptotic approximation ratio} of an algorithm $\mathcal{A}$ as the infimum $\mathcal{R} \geq 1$ such that for any input, $\mathcal{A} \leq \mathcal{R}\cdot\mathsf{OPT}+c$, where $c$ is independent of the input and $\mathsf{OPT}$ is the cost of an optimal solution for the same input. If we enforce $c = 0$, $\mathcal{R}$ is called the absolute approximation ratio or the {\em approximation ratio}. An \emph{asymptotic polynomial time approximation scheme} is a family of approximation algorithms such that for every $\eps >0$, the family contains a polynomial time algorithm with an asymptotic approximation ratio of $\oeps$. We abbreviate asymptotic polynomial time approximation scheme by APTAS (also called an asymptotic PTAS). An \emph{asymptotic fully polynomial time approximation scheme} (AFPTAS) is an APTAS whose time complexity is polynomial not only in the length of the binary encoding of the input but also in $1/\eps$. If the scheme satisfies the definition above with $c = 0$, stronger results are obtained, namely, PTASs and FPTASs, respectively. An efficient polynomial time approximation scheme (EPTAS) is a PTAS whose time complexity is upper bounded by a function of the form ${p(<I>)\cdot f(\ieps)}$ where $f$ is some computable (not necessarily polynomial) function and $p(<I>)$ is a polynomial of the length of the (binary) encoding of the input. The difference between an EPTAS and an FPTAS is that an FPTAS is defined like an EPTAS, with the added requirement that $f$ must be upper bounded by a polynomial in $\ieps$. The notion of EPTAS is a modern one, motivated in the fixed parameterized tractable (FPT) community (see e.g. \cite{cesati1997efficiency,downey_book,flum_book}).

\subsubsection{Related results.}
Exact optimization algorithms for \jai{} were developed in \cite{liao2007minimizing}. They consider short-term and long-term scheduling horizons. Their short-term case is exactly the cardinality constrained scheduling problem for which they detail an exact algorithm with time complexity of $O(2^n)$ based on branch and bound.  However, their long-term horizon case is the problem we named \jai.  For \jai\ they detail an exact algorithm with time complexity of $O(2^n)$ again based on branch and bound. 

The cardinality constrained scheduling problem on identical machines is the problem where at most $k$ jobs can be scheduled on each machine. \cite{dell2001} details lower and upper bounds for this problem. \cite{woeginger2005comment} presented an FPTAS when the number of machines is fixed to be $2$ and this FPTAS can be easily generalized to an arbitrary constant number of machines. \cite{zhang2009approximating} gave a $3$-approximation algorithm. Most important to our exposition  \cite{chen2016efficient} gave an EPTAS. See also \cite{epstein2022cardinality} for additional results for this problem.
Minimizing the makespan on identical machines is one of the most studied problems in the scheduling literature.  Here we state that Hochbaum and Shmoys presented an EPTAS for the problem (see \cite{HS87,HocBook}). 

Our problem has similar features to other scheduling models in which the machines become unavailable and available later on. For example, consider the scheduling on semi-identical processors in \cite{schmidt1984scheduling}. Semi-identical processors are processors that have different intervals of availability, but in this machine model the availability periods of each machine are given in advance as part of the input and do not depend on the decisions of the algorithm.  This feature saying that the unavailability periods are defined as part of the input and not as a result of the solution chosen by the algorithm repeats in other models as well.  See e.g.  \cite{kellerer1998algorithms,ng2004fptas,he2006improved,ji2007single,sheen2008optimal,ma2010survey,yu2014single}.

As detailed above \bpu{} generalizes the cardinality constrained bin packing problem.  \cite{kellerer1999cardinality}  presents an approximation algorithm that runs in $O(n\log^2 n)$ time with an asymptotic approximation ratio of $3/2$ for this problem.   \cite{caprara2003approximation} presents an APTAS for the cardinality constrained bin packing. \cite{epstein2010afptas} presents an AFPTAS for cardinality constrained bin packing whose time complexity was improved in \cite{jansen2019approximation}.  For standard bin packing (without cardinality constraint) we refer to \cite{karmarkar1982efficient, fernandez1981bin} for an AFPTAS.  \cite{caprara2000approximation} presents approximation algorithms for the knapsack problem with cardinality constraint, and we will use it in our scheme.

\subsubsection{Paper outlines.} In Section \ref{sch_var} we exhibit an EPTAS for \jai{}. Since the problem is strongly NP-hard, it does not have an FPTAS (unless P=NP) and an EPTAS is the best achievable scheme. For \bpu{}, we present an AFPTAS in Section \ref{bin_var}. Once again \bpu\ is strongly NP-hard so we do not expect to get an FPTAS for the problem. 

\section{An EPTAS for \jai}\label{sch_var}
\paragraph*{An overview of our scheme.}
We present an EPTAS for \jai{}. We simplify the input instance by limiting the number of job sizes via rounding. To search for the value of the optimal makespan of the rounded instance we use binary search. At each iteration of the binary search, for some search value of the makespan, we either output a feasible solution of makespan at most $\oeps$ times the search value or we declare that there is no feasible solution of makespan at most the search value.  This is accomplished by the following. We classify the jobs, based on the current search value, into small and large. In order to deal with small jobs and to help develop the mixed-integer linear program (MILP), we reformulate the problem. Based on this reformulation and the search value, we sketch partial information on the optimal solution in the form of configurations. The set of configurations is an input to the MILP that we formulate. The MILP checks if a feasible solution for the reformulated problem exists for that search value and if so, it finds one such solution. As a result of the binary search if a solution exists, we reduce the search value, else we increase it and repeat the above procedure. The smallest value of the makespan of a feasible solution that we get during the binary search is converted to the output solution for \jai{}.

Since we also need to consider cases where the problem differs from the cardinality constrained scheduling problem, we use a different method to pack the small jobs while creating the final schedule. To do that we use the reformulation of the problem.

\subsection{Preprocessing}

We first round the input instance, namely, the job sizes. Then, we search for the optimal value of the makespan of the rounded instance using binary search. Based on the current search value, we classify the job sizes. 

\subsubsection*{Rounding the input instance.}
We round up the sizes of jobs to the next integer power of $\oeps$. If the size of a job is $p_j$, then the rounded up size of the job is $p'_j = (\oeps)^{\left\lceil\log_{\oeps} p_j \right\rceil}$. Let the instance after rounding be denoted by $I'$. This rounding guarantees that $p_j \leq p'_j\leq (\oeps)p_j$ which we use in the proof of the following standard lemma.

\begin{lemma}\label{rounding}
	Any feasible schedule $\sig$, with makespan $C$, to  instance $I$ is a feasible schedule to instance $I'$ with makespan at most $(\oeps)C$. In the other direction, any feasible schedule $\sig'$ to instance $I'$, with makespan $C'$, is a feasible schedule to instance $I$ with makespan at most $C'$.
\end{lemma}
\begin{proof}
	Let $\sig$ be a feasible solution to $I$ with makespan $C$. Round up the sizes of all jobs to the nearest integer power of $\oeps$ and let $\sig'$ be the same solution when considering the input $I'$. Let $\sig_i$ and  $\sig'_i$ be the partitions $i$ in $\sig$ and $\sig'$, respectively (i.e., the set of jobs assigned to machine $i$ in the two solutions), and so $\sig_i = \sig'_i$. If the load of $i$ in $\sig$ was $L_i$, the new load of $i$ in $\sig'$ is
	\begin{align*}
		L'_i = \sum_{j\in\sig'_i}p'_j + \id\left\lfloor \frac{|\sig'_i|-1}{k}\right\rfloor \leq \sum_{j\in\sig_i}(\oeps)p_j + \id\left\lfloor \frac{|\sig_i|-1}{k}\right\rfloor \leq (\oeps)L_i \leq (\oeps) C.
	\end{align*} The makespan of $\sig'$ is at most $(\oeps)C$ since the last inequality holds for every machine $i$.

	For the other direction, consider a feasible schedule $\sig'$ to instance $I'$ of makespan $C'$. Round down the size of each job in $\sig'$ to its original size, and denote the solution by $\sig$. Let $\sig_i$ and  $\sig'_i$ be the partitions $i$ in $\sig$ and $\sig'$, respectively, and $\sig_i = \sig'_i$. If the load of $i$ in $\sig'$ was $L'_i$, the new load of $i$ in $\sig$ is
	\begin{align*}
		L_i = \sum_{j\in\sig_i}p_j + \id\left\lfloor \frac{|\sig_i|-1}{k}\right\rfloor \leq \sum_{j\in\sig'_i}p'_j + \id\left\lfloor \frac{|\sig'_i|-1}{k}\right\rfloor = L'_i \leq C'.
	\end{align*}
	Thus, the makespan of $\sig$ is at most $C'$ since the last inequality holds for all $i$.
	\qed\end{proof}

As a result of Lemma \ref{rounding}, for the remaining discussion we abuse the notation $I$ and $p_j$ to denote the rounded instance and the rounded processing time of job $j$, respectively.

\subsubsection{Binary search.}
We find an approximated value of the optimal makespan of the rounded instance as an integer power of $\oeps$. If the optimal makespan of the rounded instance is $C$, then the rounded up value is $(\oeps)^{\left\lceil\log_{(\oeps)}C\right\rceil}$. A lower bound on the optimal makespan can be obtained as $L=\max\{p_{max}, \sum_{j\in\jobs}p_j/m\}$ where $p_{max}=\max_{j\in \jobs} p_j$. This lower bound can be improved by first applying the EPTAS of \cite{chen2016efficient}. The EPTAS of \cite{chen2016efficient} for cardinality constrained scheduling is applied to the input instance $I$ for \jai{} with cardinality bound of $k$. If the EPTAS fails to provide a solution, then the number of jobs is strictly greater than $mk$ and on at least one machine the load is at least $\id$ and so $L \geq \id$. Else, we get a solution and we will compare ourselves to this solution as well. The makespan of any other schedule that we need to consider lies in the interval $\left[\max\{L,\id\} , \sum_{j\in\jobs}p_j + \id\left\lfloor (n-1)/k \right\rfloor\right]$. The upper bound holds by considering the solution where we assign all jobs to a common machine. The number of distinct values possible for the rounded up value of the makespan is at most
\begin{align*}
	\log_{\oeps}\left(\frac{\sum_{j\in\jobs}p_j}{L} + \frac{\id}{\id}\left\lfloor \frac{(n-1)}{k} \right\rfloor\right)+2\leq \log_{\oeps}\left(\frac{\sum_{j\in\jobs}p_j}{p_{max}} + \frac{(n-1)}{k}\right)+2 \leq \log_{\oeps}\left(2n\right)+2,
\end{align*}
and is strongly polynomial. The final inequality is because $k\geq 1$ and $\sum_{j\in\jobs}p_j\leq np_{max}$.

Instead of going over all the possible values, we use binary search on this set of possible values to obtain the rounded value of the makespan of the optimal schedule of $I$. The current search value for the rounded makespan of the rounded instance is denoted by $\opt$. Then, we approximately check if a solution exists with makespan at most $\opt$ (see below for the meaning of approximately check). If a solution does not exist, we increase the search value, and otherwise we decrease it. The number of steps required by the binary search is at most $\log_{2}\left(\log_{\oeps}\left(2n\right)+2\right)$.

\subsubsection{Job classification.}
We classify the jobs into two categories: small and large. A job is \emph{small} if the size of the job is strictly less than $\eps\opt$, else it is \emph{large}. Let $\sma$ be the set of small jobs, and let $\lar$ be the set of large jobs. Let $\larsiz$ be the set of distinct sizes of large jobs. Since the large job sizes lie in the interval $[\eps\opt, \opt]$ and the input is rounded, we have $$|\larsiz|\leq \log_{\oeps} \opt - \left\lceil\log_{\oeps} \eps\opt\right\rceil+ 1 \leq \log_{(\oeps)} 1/\eps + 1,$$ since $\opt$ is an integer power of $\oeps$. Furthermore, the number of large jobs that can be assigned to a common machine is at most
\begin{align}\label{large_per_machine}
	\opt/(\eps\opt) = 1/\eps,
\end{align}
as otherwise the upper bound on the makespan is violated.

\subsection{Reformulation}
We reformulate \jai{} to make it easier to deal with small jobs. The reformulated problem is denoted as \jaii{}. The input instance of \jaii{} is the same as the input instance of \jai{}, i.e., a set $\jobs$ of $n$ jobs where each job $j$ has a size $p_j \in\nre$, a set $\machines$ of $m$ parallel identical machines, a positive integer $k$, and the duration of the unavailability periods $\id\in\nre$.

Next, we define the concept of early and late jobs of a set of jobs. A job $j$ in some set of jobs (assigned to machine) $i$ is \emph{early} if it is among the first $k/\eps$ jobs in $i$, else it is \emph{late}. We use $\early_i$ and $ \late_i $ to denote the set of early and late jobs of the set $i$, respectively.

A \emph{feasible solution} for \jaii\ is a partition of $\jobs$ into $m$ subsets each of which is sorted in a non-increasing order (breaking ties arbitrarily). Each subset $i$ will have a partition into early and late jobs of $i$ denoted by $ \early_i $ and $\late_i$, respectively, based on its sorted list. \jaii{} requires a modified unavailability requirement that is different from the one of \jai{}. In \jaii{}, the unavailability requirement needs to be satisfied  only by the early jobs. That is, if the number of jobs assigned to machine $i$ is at most $k/\eps$, the unavailability requirement is identical to the one of \jai{}, else the jobs in $i$ will require exactly $\ieps$ periods of idle time (standing only for the $k/\eps$ early jobs). But, the processing time of each late job on such machine is increased by $\id/k$. If job $j$ is early, the processing time of $j$ is its size $p_j$, else the processing time is its modified size defined as $\tp_j = p_j + \id/k$. Thus, the \emph{load} of machine $i$ whose assigned set of jobs is $\sigma_i$ in a solution to $\jaii{}$ is defined as
\begin{align*}
	L_i = \begin{cases}
		\sum_{j\in\early_i}p_j  + \left\lfloor\dfrac{|\early_i|-1}{k}\right\rfloor U, & \text{if } |\sig_i| \leq k/\eps\\
		\sum_{j\in\early_i}p_j  + \sum_{j\in\late_i}\tp_j+ \dfrac{\id}{\eps}, & \text{if } |\sig_i| > k/\eps
	\end{cases}
\end{align*}
where $\late_i$ is the set of late jobs and $\early_i$ is the set of early jobs in $i$, respectively. The \emph{makespan} of a solution is the largest load among all machines. The objective function of \jaii{} is to minimize the makespan. All the large jobs are  early (because of the sorting of the jobs and from \ref{large_per_machine}) and only small jobs are perhaps late.

To give an intuitive explanation of this reformulation consider a feasible schedule to \jaii{}. The processing time of a late job $j$ is $\tp_j = p_j+{\id}/{k}$ where ${\id}/{k}$ is a (fractional) idle item coupled with the processing time of $j$. When $k$ such small jobs are scheduled together, we get a total of $\id$ idle time by rearranging the processing of the jobs and the idle time, satisfying the machine unavailability requirement after processing $k$ jobs. The following lemma shows that this reformulation does not cause a large loss to the approximation ratio.

\begin{lemma}\label{reformulation}
	A feasible solution $\sig$ with makespan at most $\opt$ of \jai{} is a feasible solution of \jaii{} with makespan at most $(1+\eps)\opt$.
\end{lemma}
\begin{proof}
	Consider a feasible solution $\sig$ of \jai{} with makespan at most $\opt$. Consider the same solution for \jaii{}. The solution is feasible to \jaii{} (after increasing the processing time of late jobs). Let $M$ be the set of partitions with at least one late job i.e., $M=\{ i: |i| > k/\eps\}$. We need to consider only machines in $M$ because the modification to the load occurs only to  these machines. For $i\in M$, let $L'_i$ be the load of $i$ in \jaii{} and $L_i$ be the load of $i$ according to the definition of load in \jai{}. Let $\early_i$ and $\late_i$ be the set of early jobs and late jobs of $i$, respectively, and since $i\in M$, $|\early_i|/k=\ieps$. We get that
	\begin{align*}
		L'_i - L_i &=
		\sum_{j\in\early_i}p_j +\sum_{j\in\late_i}\tp_j+\frac{\id}{\eps} - \left(\sum_{j\in \sigma_i}p_j+\left\lfloor \frac{|\sigma_i|-1}{k} \right\rfloor\id \right)\\
		&=
		\sum_{j\in\early_i}p_j +\sum_{j\in\late_i}\left(p_j+\frac{\id}{k}\right)+\frac{\id}{\eps} - \left(\sum_{j\in \sigma_i}p_j+\left\lfloor \frac{|\sigma_i|-1}{k} \right\rfloor\id \right)\\
		&=\frac{|\late_i|}{k}\id -\left\lfloor \frac{|\late_i|-1}{k} \right\rfloor\id \leq \id.
	\end{align*}

	Since $L_i \geq \id/\eps$, the new load is at most $(1+\eps)L_i$ for all $i\in M$ and the makespan is at most $(\oeps)\opt$.
\qed\end{proof}

Observe that a  feasible solution $\sig$ with makespan at most $\alpha$ of \jaii{} is a feasible solution of \jai{} with makespan at most $\alpha$.  However, we will not use this claim and so we omit its trivial proof.
Next, we use a mixed-integer linear program to obtain a feasible solution to \jai{}.  We will require that if there is a solution to \jai\ of makespan at most $\opt$ and thus also a solution to $\jaii$ of makespan at most $\opt$, then the algorithm will output a feasible solution to $\jai$ of makespan at most $(\oeps)\cdot \opt$. 

Henceforth, let $\opt$ denote the approximated optimal makespan of \jaii{}. 

\subsection{The mixed-integer linear program (MILP)}
Given a search value of the optimal makespan of \jaii\ the MILP allows us to check if a feasible solution exists to \jaii{} with makespan at most $1+\eps$ times the search value or there is no feasible solution of makespan at most the search value.  With a slight abuse of notation we let $\opt$ denote the search value of \jaii. The MILP uses partial information on the feasible assignments to one machine in the form of configurations. The MILP can be solved in polynomial time using the approach of \cite{lenstra1983,kannan1983}. The assignment of the jobs will be done by using the solution of the MILP together with some post-processing algorithm (see below). Notice that we are dealing with the problem \jaii{}.

\subsubsection{Configurations.}
A configuration $c$ encodes the information of a set of jobs assigned to a common machine $i$. It is a vector of length $|\larsiz|+4$. The components of a configuration $c$ are as follows:
\begin{outline}[enumerate]
	\1 The first $|\larsiz|$ components store the number of large jobs of a particular size in $i$. For each $\ell\in\larsiz$, we have a component $\alpha_{c\ell}$ that stores the number of large jobs of size $\ell$ in $i$. Recall that all large jobs are early jobs.

	\1 The next component denoted by $\beta_{c}$ stores the total size of small jobs that are early jobs (in $i$) rounded down to the next integer multiple of $\eps\opt$ i.e., the total size of small jobs that are early (in $i$) is in the interval $[\beta_{c}\cdot \eps\opt, (\beta_{c}+1)\cdot\eps\opt)$.

	\1 The next component denoted by $\gamma_{c}$ stores the total modified size of small jobs  that are late jobs (in $i$) rounded down to the next integer multiple of $\eps\opt$ i.e., the total modified size of small jobs that are late (in $i$) is in the interval $[\gamma_{c}\cdot\eps\opt, (\gamma_{c}+1)\cdot\eps\opt)$.

	\1 The next component denoted by $\gamma'_c$ is a binary value and is $1$ if there is at least one late job, and otherwise it is $0$.

	\1 The last component denoted by $\delta_c$ stores the total number of early jobs rounded up to an integer multiple of $k$. That is, the total number of jobs assigned as early jobs lies in the interval $[(\delta_c-1)\cdot k +1 , \delta_c \cdot k]$.
\end{outline}

The load of a configuration $c$ is defined as $$L_c = \sum_{\ell\in\larsiz}\ell\cdot \alpha_{c\ell} + \beta_c\eps\opt +\gamma_c'\gamma_{c}\eps\opt + \id(\delta_c-1+\gamma_c').$$ A configuration $c$ is \emph{feasible} if the load of the configuration is at most $\opt$. Let $\cons$ be the set of feasible configurations.

\begin{lemma}\label{int-var}
	The number of feasible configurations is at most $2\left(\ieps+1\right)^{\left(\log_{(\oeps)} \ieps + 4\right)}$.
\end{lemma}
\begin{proof}
	Consider a feasible configuration $c$ and by its feasibility, the load of $c$ is at most $\opt$. Since $\opt \geq L_c\geq \sum_{\ell\in\larsiz}\alpha_{c\ell}\ell \geq \eps\opt\sum_{\ell\in\larsiz}\alpha_{c\ell}$, the number of distinct values that each $ \alpha $ can have is at most $\ieps+1$. The numbers of distinct possible values for the $\beta_c$ and $\gamma_c$ are at most $\ieps+1$ since the total size of small jobs (early and late) lies in $[0,\opt]$. The number of early jobs is at most $k/\eps$ and the number of distinct possible values for $\delta_c$ is at most $\ieps+1$. The number of feasible configurations, $|\cons|$, is at most $2\left(\ieps+1\right)^{|\larsiz| + 3} \leq 2\left(\ieps+1\right)^{\left(\log_{(\oeps)} \ieps + 4\right)}$ and that is a constant when $\eps$ is fixed.
\qed\end{proof}

\subsubsection{Decision variables of the MILP.}

\paragraph*{Configuration counters.} We have a configuration counter variable $x_c$ associated with each $c\in\cons$. The configuration counter variable $x_c$ corresponding to $c$ counts the number of machines whose assigned set of jobs is represented by the feasible configuration $c$.  All these variables will be forced to be integers in the MILP.  By Lemma \ref{int-var}, the number of these variables does not depend on the input encoding length.

\paragraph*{Assignment variables for small jobs to configurations.} We have two types of assignment variables for small jobs to configurations. The first type is denoted by $\textbf{y}$. A variable $y_{jc}$ is $1$ if a small job $j$ is an early job in the set of jobs represented by configuration $c$, and otherwise it equals $0$. The second type is denoted by $\textbf{z}$. A variable $z_{jc}$ is $1$ if a small job $j$ is a late job in the set of jobs represented by configuration $c$, and otherwise it is $0$.  All these variables (of both types) are allowed to be fractional in the MILP.  The number of these variables is polynomially bounded in the input encoding length.

\subsubsection{Constraints.}

We provide an intuitive explanation along with the constraints of the MILP.

The total configuration counter chosen should be equal to the number of machines which is $m$.
\begin{align}\label{cons_machines}
	\sum_{c\in\cons} x_c = m.
\end{align}
All small jobs should be assigned completely (either as early jobs or as late jobs).
\begin{align}\label{small_1}
	\sum_{c\in\cons}\left(y_{jc} + z_{jc}\right) =1,\ \ \forall j\in\sma.
\end{align}
The total size of small jobs assigned as early jobs to configuration $c$ should satisfy the space requirement defined by $c$.
\begin{align}\label{small_up_to}
	\sum_{j\in \sma} y_{jc}p_j \leq x_c(\beta_{c}+1)\eps\opt, \ \  \forall c\in \cons.
\end{align}
Similarly, the total size of small jobs assigned as late jobs should satisfy the space requirement, with the modified size, defined by the configuration.
\begin{align}\label{small_after}
	\sum_{j\in \sma} z_{jc}\tp_j \leq x_c\gamma'_c(\gamma_{c}+1)\eps\opt, \ \  \forall c\in \cons.
\end{align}
We do not have assignment variables for assigning large jobs. Instead, we have \emph{places} for large jobs in the configurations. The next family of constraints enforces that there are enough places for all the large jobs to be assigned.
\begin{align}\label{large_places}
	\sum_{c\in\cons}\alpha_{c\ell}x_c = |\lar_{\ell}|, \ \ \forall \ell\in\larsiz,
\end{align}
where $\lar_{\ell}$ is the set of large jobs of size $\ell.$ The total number of jobs assigned to configurations as early jobs is restricted by the configurations.
\begin{align}\label{delta_constraint}
	\sum_{\ell\in\larsiz}\alpha_{c\ell}x_c + \sum_{j\in\sma}y_{jc} \leq x_c\delta_ck, \ \  \forall c\in\cons
\end{align}
Finally, we have the integrality and non-negativity constraints for the variables.
\begin{align}
	&x_c\in\nint , \ \ \ \forall c\in\cons         \\
	0 \leq &y_{jc} \leq 1, \ \ \ \forall c\in\cons \ \  \forall j\in \sma \\
	0 \leq &z_{jc} \leq 1 , \ \ \ \forall c\in\cons \ \  \forall j\in \sma .
\end{align}

Next, we show that there exists a MILP solution if \jaii{} has a feasible solution and the solution satisfies the above constraints.

\begin{theorem}
	If there exists a feasible solution to \jaii{} of makespan at most $\opt$, then there exists a feasible MILP solution.
\end{theorem}
\begin{proof}
	Consider a feasible solution $\sig$ to \jaii{}; $\sig$ is a partition of \jobs{} into $m$ subsets. For a partition $i$ that is a job set of machine $i$ according to $\sig$, let $ \early_i $ and $ \late_i $ be the set of early and late jobs of $i$, respectively. The basic flow of the proof is that first, based on $\sig$,  we create configurations. Then, we show that the generated configurations are feasible.  Afterwards, we define values for the MILP decision variables. Finally, we show that the vector of values of the decision variables is a feasible solution for the MILP, that is, it satisfies every constraint of the MILP. Recall that $\cons$ is the set of feasible configurations.

	\paragraph*{Generating a configuration for each machine.} Perform the following operations for each machine $i$ in order to define a configuration $c(i)$ corresponding to machine $i$. For each $\ell\in\larsiz$, $\alpha_{c(i)\ell}$ is the number of large jobs of size $\ell$ in $i$. The $\delta_{c(i)}$ value of $c(i)$ is equal to the number of early jobs in $i$ rounded up to an integer multiple of $k$. That is, $\delta_{c(i)}  = \left\lceil{|\early_i|}/{k}\right\rceil.$ Round down the total size of early small jobs in $i$ to the nearest integer multiple of $\eps\opt$ and store the integer multiple as the $\beta_{c(i)}$ value. That is, $\beta_{c(i)} = \left\lfloor {\sum_{j\in \early_i\backslash\lar}p_j}/{\eps\opt}\right\rfloor.$ If there are no jobs processed on machine $i$ as late jobs, that is, if $\late_i=\emptyset$, set $\gamma_{c(i)}$ and $\gamma'_{c(i)}$ to be $0$. Else, do the following. Round down the total modified size of the late jobs to the nearest integer multiple of $\eps\opt$ and store the integer multiple as the $\gamma_{c(i)}$. That is, $\gamma_{c(i)} = \left\lfloor{\sum_{j\in\late_i}\tp_j}/{\eps\opt}\right\rfloor.$ In this case set $\gamma'_{c(i)} = 1$. The multiset of all configurations generated from all machines is denoted by $\mathcal{C}$.

	\paragraph*{Checking feasibility of generated configurations.} Next, we check the feasibility of the generated configurations. The criterion for a configuration to be feasible is that the load of the configuration should be at most $\opt$. Consider a configuration $c(i)$ generated from machine $i$. The load of $i$ according to its definition in \jaii\ is
	\[
	L_i =
	\begin{cases}
		\sum_{j\in\early_i}p_j  + \left\lfloor\frac{|\early_i|-1}{k}\right\rfloor U, & \text{if } |i| \leq k/\eps\\
		\sum_{j\in\early_i}p_j  + \sum_{j\in\late_i}\tp_j+ \frac{\id}{\eps}, & \text{if } |i| > k/\eps
	\end{cases}
	\]
	where $\early_i$ is the set of early jobs and $\late_i$ is the set of late jobs of $i$, respectively. Also, $i = \early_i\cup\late_i$. First assume that $\late_i=\emptyset$. Then, by definition, $\gamma_{c(i)}'=0$ and $i = \early_i$. The load of the configuration $c(i)$ in this case is
	\begin{align*}
		L_{c(i)} &= \sum_{\ell\in\larsiz}\alpha_{c(i)\ell}\ell + \beta_{c(i)}\eps\opt + \id(\delta_c-1)\\
		&= \sum_{\ell\in\larsiz}\alpha_{c(i)\ell}\ell +\left\lfloor \frac{\sum_{j\in\early_i\backslash\lar}p_j}{\eps\opt}\right\rfloor\eps\opt + \id\left(\left\lceil\frac{|\early_i|}{k}\right\rceil-1\right)\\
		&\leq \sum_{j\in i}p_j + \left\lfloor\frac{|\early_i|-1}{k}\right\rfloor\id = L_{i},
	\end{align*}
	and we conclude that in this case $L_{c(i)}\leq L_i$.
Next, consider the case where $\late_i\neq \emptyset$. Then, $\gamma_{c(i)}'=1$  and $|\early_i|=k/\eps$. The load of the configuration $c(i)$ in this case is
	\begin{align*}
		L_{c(i)} &= \sum_{\ell\in\larsiz}\alpha_{c(i)\ell}\ell + \beta_{c(i)}\eps\opt + \gamma'_c\gamma_c\eps\opt + \id\cdot \delta_c\\
		&= \sum_{\ell\in\larsiz}\alpha_{c(i)\ell}\ell +\left\lfloor \frac{\sum_{j\in\early_i\backslash\lar}p_j}{\eps\opt}\right\rfloor\eps\opt + \left\lfloor \frac{\sum_{j\in\late_i}\tp_j}{\eps\opt}\right\rfloor\eps\opt +\id\left(\left\lceil\frac{|\early_i|}{k}\right\rceil\right)\\
		&\leq \sum_{j\in\early_i}p_j +\sum_{j\in\late_i}\tp_j+ \frac{\id}{\eps}= L_{i}.
	\end{align*}
	Thus, in both cases $L_{c(i)}\leq L_i$.
	Since $\sig$ was feasible, $L_{c(i)}\leq L_i \leq \opt$. Thus, all configurations are feasible.

	\paragraph*{Defining values for the decision variables.} Next, we exhibit values for the MILP decision variables from the schedule and the configurations. For every $c\in\cons$, let $x_c$ be the number of times $c$ appear in $\mathcal{C}$. Initialize all the $\textbf{y}$ and $\textbf{z}$ values to $0$. For each small job $j$, identify the machine $i$ on which it is assigned. If $j$ is processed on $i$ as an early job, set $y_{jc(i)}=1$, else set $z_{jc(i)}=1$, where $c(i)$ is the configuration generated for machine $i$.

	\paragraph*{Checking feasibility of the values for the decision variables.} Next, we check the feasibility of the values we have exhibited by checking whether the MILP constraints are satisfied by them. Consider a configuration $c\in\cons$, and let $m_c$ be the set of machines whose generated configuration is $c$, thus $|m_c| = x_c$. From the definition of $\textbf{x}$ values and since each machine has exactly one configuration corresponding to it, we have $\sum_{c\in\cons}x_c = |\mathcal{C}| = m$, and Constraint \eqref{cons_machines} is satisfied. Constraint \eqref{small_1} is satisfied because $\sig$ is a partition of the job set, and from the definitions of $\textbf{y}$ and $\textbf{z}$. Since $\beta_{c} = \left\lfloor \sum_{j\in\early_i\backslash\lar}p_j/(\eps\opt)\right\rfloor$, we obtain $\sum_{j\in\early_i\backslash\lar}p_j < (\beta_{c}+1)\eps\opt,\ \forall i\in m_c$. $\sum_{j\in\sma}y_{jc}p_j$ gives the total size of all early small jobs in $m_c$. Considering the $x_c$ copies of $c$, we have $\sum_{j\in\sma}y_{jc}p_j = \sum_{i\in m_c} \sum_{j\in\early_i\backslash\lar}p_j < x_{c}(\beta_{c}+1)\eps\opt$, and Constraint \eqref{small_up_to} is satisfied. Since $\gamma_{c} = \left\lfloor \sum_{j\in\late_i}\tp_j/(\eps\opt)\right\rfloor$, $\sum_{j\in\late_i}\tp_j < \gamma'_{c}(\gamma_{c}+1)\eps\opt,\ \forall i\in m_c$. $\sum_{j\in\sma}z_{jc}\tp_j$ gives the total modified size of all late jobs in  $m_c$. Considering the $x_c$ copies of $c$, we conclude that $\sum_{j\in\sma}z_{jc}\tp_j = \sum_{i\in m_c} \sum_{j\in\late_i}\tp_j < x_{c}\gamma'_c(\gamma_{c}+1)\eps\opt$, and Constraint \eqref{small_after} is satisfied. From the definition of the configurations and $\textbf{x}$, for each size of large jobs, the number of places available for large jobs of this size is equal to the number of large jobs (of this size), and Constraint \eqref{large_places} is satisfied. Since $\delta_{c} = \left\lceil |\early_i|/{k} \right\rceil$, we conclude that $ |\early_i| \leq \delta_{c}k,\ \forall i\in m_c$. The total number of jobs that are early in $m_c$ is $x_c\sum_{\ell\in\larsiz}\alpha_{c\ell} + \sum_{j\in\sma}y_{jc}$. Thus, $x_c\sum_{\ell\in\larsiz}\alpha_{c\ell} + \sum_{j\in\sma}y_{jc} = \sum_{i\in m_c}|\early_i| \leq x_c\delta_ck$, and Constraint \eqref{delta_constraint} is satisfied. The non-negativity and integrality constraints are satisfied by the definitions of $\textbf{x}, \textbf{y}$, and $\textbf{z}$ values. Thus, $(\textbf{x}, \textbf{y}, \textbf{z})$ is a feasible MILP solution.
\qed\end{proof}

\subsection{Converting the MILP solution into the output of the scheme}
This section details how to transform a feasible MILP solution along with the set of feasible configurations to a feasible solution to \jai{} in time polynomial in the input encoding length.

\paragraph*{Best-Fit schedule.}
Next, we state the analysis of the Best-Fit schedule for scheduling with cardinality constraint that was shown  in \cite{chen2016efficient}. We will use it as a sub-routine later in our algorithm. Let $c_{i}\in\nint$ be the capacity bound on machine $i$, and $\pi_j\in\nre$ represent the size of a job $j$.

\begin{lemma}[{\cite[Lemma~2]{chen2016efficient}}]\label{best_fit}
	Assume that $c_i\in\nint$ and $t_i\in\nre$ for all $i$. If there is a feasible solution $\textbf{u}$ for the following linear system
	\begin{align*}\tag{LP}\label{LP}
		\sum_{j=1}^{n} \pi_ju_{ij} \leq t_i           & ,\ 1\leq i\leq m                    \\
		\sum_{j=1}^{n} u_{ij} \leq c_i              & ,\ 1\leq i\leq m                    \\
		\sum_{i=1}^{m} u_{ij} = 1                   & ,\ 1\leq j\leq n                    \\
		0\leq u_{ij} \leq 1                         & , \ 1\leq j\leq n,\ \ 1\leq i\leq m \ ,
		\intertext{then an integer solution $\textbf{u}'$ satisfying:}
		\sum_{j=1}^{n} \pi_ju'_{ij} \leq t_i + \pi_{max} & ,\ 1\leq i\leq m                    \\
		\sum_{j=1}^{n} u'_{ij} \leq c_i              & ,\ 1\leq i\leq m                    \\
		\sum_{i=1}^{m} u'_{ij} = 1                   & ,\ 1\leq j\leq n                    \\
		u'_{ij} \in \{0,1\}                          & , \ 1\leq j\leq n,\ \ 1\leq i\leq m
	\end{align*}
	could be obtained in $O(n\log n)$ time, where $\pi_{\max}=\max_{j}\{\pi_j\}$.
\end{lemma}

Next, we show that $\cons$ and $(\textbf{x}, \textbf{y}, \textbf{z})$ gives us the necessary information to create a feasible solution to \jai{} and the makespan of the corresponding schedule is at most $(1+3\eps)\opt$.

\begin{theorem}\label{theorem_2}
	If there exists a solution to the MILP for the rounded instance $I$ with the given upper bound on the optimal makespan value $\opt$, then there exists a feasible schedule to instance $I$ for \jai{} of makespan at most $(1+3\eps)\opt$ that can be found in polynomial time.
\end{theorem}
\begin{proof}
	Let $\opt$, $\cons$, and $(\textbf{x}, \textbf{y}, \textbf{z})$ be the current approximated value of the optimal makespan, the set of feasible configurations, and the MILP solution, respectively. This information will be converted to a feasible schedule for \jai{} as follows. First, we assign configurations to machines based on the MILP solution. For each configuration $c\in\cons$, we assign $x_c$ machines with $c$. Every machine $i\in \machines$ is assigned exactly one configuration. By Constraint \eqref{cons_machines}, such assignment can be defined in linear time.

	\paragraph*{Assigning large jobs.} For each $\ell\in\larsiz$ and $i\in \machines$ assign $\alpha_{c(i)\ell}$ large jobs of size $\ell$ to machine $i$ whose assigned configuration is $c(i)$. From Constraint \eqref{large_places}, all large jobs will be assigned using this rule.

	\paragraph*{Assigning small jobs.} The remaining jobs are small jobs. We convert the assignment variables for small jobs to configurations to assignment variables for small jobs to machines as follows. Do the following for each machine. Consider a machine $i$ and let $c(i)$ be the configuration assigned to $i$.  In an arbitrary order of $j\in \sma$ do the following. We assign $y_{jc(i)}/x_{c(i)}$ fraction of small job $j$ to $i$ as an early job and assign $z_{jc(i)}/x_{c(i)}$ fraction of small job $j$ to $i$ as a late job. The total space occupied by the fractions of the small jobs assigned to $i$  as early jobs is $\sum_{j\in\sma}\left(y_{j{c(i)}}/x_{c(i)}\right)p_j$ and from Constraint \eqref{small_up_to} it is at most $(\beta_{c(i)}+1)\eps\opt$. Similarly, the total modified size of fractions of the small jobs assigned to $i$ as late jobs is $\sum_{j\in\sma}\left(z_{j{c(i)}}/x_{c(i)}\right)\tp_j$ and from Constraint \eqref{small_after} it is at most $(\gamma_{c(i)}+1)\eps\opt$. By Constraint \eqref{small_1}, all small jobs will be assigned fractionally to machines. We separate the idle time from all the late jobs, i.e., the processing time of a late job will be equal to its value in the rounded instance of \jai\ and not as in \jaii. The idle time can be rearranged to get the necessary idle time required for the unavailability requirement as a result of Constraint \eqref{delta_constraint}.

	Based on these assignments we create a new vector of assignment variables $\textbf{u}$. If $f$ fraction of job $j$ is assigned to machine $i$, then $u_{ij}=f$. For each machine $i$, calculate the cardinality bound on the number of small jobs assigned to $i$ as $c_i = \left\lceil\sum_{j\in\sma}u_{ij}\right\rceil$. For each machine $i$, calculate the total size of small jobs assigned to machine $i$ as $t_i = \sum_{j\in\sma}u_{ij}p_j$ and notice that we use size as in the rounded instance $I$. Now we have all the necessary parameters to use \eqref{LP} and we observe that $\textbf{u}$ is feasible to \eqref{LP}. Applying Lemma \ref{best_fit} we get an integral assignment of small jobs to machines denoted by $\textbf{u}'$. Applying Lemma \ref{best_fit} increases the load of every machine by at most $\pi_{max} \leq \eps\opt$ since all jobs assigned using that lemma are small jobs.

	The number of idle time periods on each machine is not increased (beyond that encoded in the configuration) using the fact that by rounding up to the next integer, the next integer multiple of $k$ remains the same. The makespan of the schedule is at most $(1+3\eps)\opt$, due to the fact that our fractional assignment of jobs to configuration is over-using the space defined in  the $\beta$ and $\gamma$ values by at most $\eps\opt$ (each), and by the application of lemma \ref{best_fit}.
\qed\end{proof}

\begin{theorem}
	Problem \jai{} admits an EPTAS.
\end{theorem}
\begin{proof}
	All the procedures applied during the algorithm can be performed in time polynomial in the input encoding length. The makespan of the output schedule we obtain is at most $(1+O(\eps))\opt$ as a result of Lemma \ref{rounding} and \ref{reformulation} and Theorem \ref{theorem_2}.
\qed\end{proof}

\section{An AFPTAS for \bpu}\label{bin_var}
Here, we turn our attention to the dual problem, namely, \bpu.
We use $\opt$ to denote an optimal solution of the input instance $I$, and with a slight abuse of notation let $\opt$ also be the cost of the optimal solution (that is, the number of partitions or bins in the optimal solution). 

\paragraph*{Overview.}
An upper bound on the number of items in a bin is $$\cb:=\left\lfloor \frac{k}{\id} \right\rfloor + k$$ since the size of an item is non-negative.  We consider two cases based on the value of $\cb$.
The first and easier case is when $\cb \leq \ieps^2$ while the more difficult one is when $\cb > \ieps^2$.  We follow earlier schemes for cardinality constrained bin packing and  deal with each case separately.

For the first case we use linear grouping and rounding on all the items. Then, the possible partitions are enumerated using an exponential number of configurations. This exponential list of the configurations is not used directly. We use a linear program and the column-generation technique of \cite{karmarkar1982efficient} to obtain an approximate feasible solution close to the optimal solution in polynomial time (polynomial in the encoding length of the \bpu\ instance) to obtain an AFPTAS. For the second case we use linear grouping on only some items after item classification. The rest of the procedure involves using a linear program and the column generation technique to find an approximate feasible solution close to the optimal solution in polynomial time to obtain an AFPTAS.

We note that when we consider our scheme for the special case of the cardinality constrained bin packing problem, we develop a new AFPTAS for this special case.  The differences from earlier schemes arise mostly in the assignment of small items in the second case. We use a configuration linear program similar to the earlier schemes, but the conversion of the fractional assignment of small items to integral packing of small items is performed with a simple rounding of the basic solution for the linear program. Thus, our scheme can be considered as simplification and generalization of the earlier schemes for the cardinality constrained bin packing problem.

\subsection{Case I: $\cb \leq \ieps^2$}
In this case the number of items in a bin is upper bounded by a constant.  In the terminology of bin packing approximation schemes, we will define all items as large and apply linear grouping of the entire item set. Thus, in a nutshell our first step is to round the item sizes and obtain a rounded instance in which the number of distinct sizes is a constant. Then, we use a linear program to find a feasible packing by implicitly enumerating all possible ways to pack the rounded items in a bin. The enumeration is done using the standard definition of configurations.

\subsubsection{Linear grouping and rounding.}
Sort $\items$ in a non-increasing order of item sizes and then the sorted list is partitioned into sub-lists of consecutive items called \emph{classes}. We fix the number of classes to be $\ieps^3$, and the classes are denoted by $\items_i,i\in[\ieps^3]$. Some classes might be empty depending on the cardinality of $\items$. The partition into classes is based on the following criteria: $\left\lceil \eps^3|\items| \right\rceil = |\items_1| \geq |\items_2| \geq \ldots \geq |\items_{\ieps^3}| = \left\lfloor \eps^3|\items| \right\rfloor$.  Observe that the last set of conditions defines the cardinality of each class so the partition into classes is well-defined once the ordering of equal items in the sorted list is defined.
 The first class has the largest $\left\lceil \eps^3|\items| \right\rceil$ items. We remove the items in class $\items_1$ to be packed separately, and let $\items'=\items\backslash\items_1$ and the classes in $\items'$ start from $\items_2$. Then, we round the size of items in $\items'$ as follows. Let $s(\items_i)$ denote the largest item size in class $\items_i$ for all $i$. The size of each item in $\items'$ is rounded up to the largest item size in the class to which the item belongs to. That is, the size of each item in class $\items_i$ is rounded up to $s(\items_i)$. The resulting number of distinct item sizes in $\items'$ is at most $\ieps^3-1$. Let the rounded instance, with the rounded items in $\items'$, be denoted by $I'$.

\begin{lemma}\label{grouping}
	The cost of the optimal solution of the rounded instance is at most the cost of the optimal solution of the original instance, $\opt$.
\end{lemma}
\begin{proof}
	$\opt$ is the optimal solution of the original instance $I$, i.e., a partition of $\items$ into feasible bins. We partition the items in $I'$ into $\opt'\leq  \opt$ bins as follows.


 First, we identify the class in the linear grouping containing each item . For $i=1,2,\ldots,\ieps^3-1$, we place the $j$-th (rounded) item of class $\items_{i+1}$ in $\opt'$ (i.e., the $j$ item counting only the sublist of items of the class) in the position of the original $j$-th item in class $\items_i$ in $\opt$. From the linear grouping and rounding we have that:
	\begin{outline}[enumerate]
		\1 the number of items in each successive class is at most the number of items in the preceding class i.e., $|\items_i| \geq |\items_{i+1}|$, and
		\1 the size of items in each successive class is at most the size of the smallest item in the preceding class.
	\end{outline}
	Let $p'$ be a set of items packed into a common bin in $\opt'$ and let $p\in\opt$ be the same bin in $\opt$ (its item set is different though). The above two facts implies that the number of items in $p$ is at least the number of items in $p'$, and the total size of all rounded items in $p'$ is at most the total size of all items in $p$. Thus, the load of $p'$ is at most the load of $p$, and since $p$ is feasible, $p'$ is also feasible. The number of partitions in $\opt'$ is thus at most $\opt$, and the claim follows.
\qed\end{proof}

\begin{lemma}\label{1_items}
	The items in $\items_1$ can be packed into at most $\eps\opt+1$ additional feasible bins.
\end{lemma}
\begin{proof}
	We let each item in $\items_1$ be in a dedicated bin on its own. From the linear grouping, we have $|\items_1| \leq \eps^3|\items| + 1$. The number of items possible in a bin in a feasible solution is at most $\cb \leq \ieps^2$, and $\opt \geq \frac{|\items|}{\cb} \geq \eps^2|\items|$. This gives us
	$$	|\items_1| \leq \eps^3|\items| + 1 \leq \eps\opt+1.$$
\qed\end{proof}

As a result of the above lemmas, our objective now is to obtain a partition of the items in the rounded instance $I'$. This is done using a configuration linear program. Using the configuration linear program we initially obtain a fractional solution which is rounded to obtain the final integral solution for \bpu{}.

\subsubsection{The Configuration Linear Program.}
We use a configuration linear program to get the packing of the bins. A configuration represents the set of items that can be packed in a feasible bin. We use an encoding where a set of items is defined by the number of items of each distinct size. Let the distinct item sizes be denoted by $\sizes$, and recall that by the linear grouping $|\sizes|\leq \ieps^3-1$.  Let $n(z)$ be the number of items of size $z$ in the rounded instance.

\paragraph*{Configurations.} A configuration is a vector of length $|\sizes|+1$. The first $|\sizes|$ components of a configuration $c$ are denoted by $\alpha_{cz},\forall z\in\sizes$, and $\alpha_{cz}$ is the number of items of size $z$ in the item set represented by $c$. The last component is denoted by $\delta_c$, and it is the number of idle items necessary to satisfy the unavailability requirement by the item set represented by $c$. A (non-empty) configuration $c$ is \emph{feasible} if the following two conditions are satisfied. 
\begin{outline}[enumerate]
	\1 The number of idle items required by the item set represented by $c$ should depend correctly on the number of items in it. This is based on the unavailability requirement; that is,
	\begin{align*}
		\delta_c = \left\lfloor\frac{\sum_{z\in\sizes}\alpha_{cz}-1}{k}\right\rfloor.
	\end{align*}
	\1 The load of the item set represented by $c$ should be at most $1$, that is,
	\begin{align*}
		\sum_{z\in\sizes}z\alpha_{cz} + \delta_c\id\leq 1 .
	\end{align*}
\end{outline}

The number of feasible configurations is exponential in $\ieps$. Let $\cons$ denote the set of feasible configurations. Since $|\cons|$ is an exponential function of $\ieps$, our algorithm does not list all configurations in $\cons$, and this is just notation to be used below.

\paragraph*{Decision variables.}
The decision variables are a vector $\textbf{x}$ of length $|\cons|$. A component $x_c$ for some $c\in\cons$ is the number of bins whose item set is $c$. When $x_c>0$, for every $z\in \sizes$, the configuration $c$ will provide $x_c\alpha_{cz}$ places for items of size $z$.

\paragraph*{Constraints.}
We have a constraint for each distinct size of items. The constraint is that the selection of the configurations should be such that there are enough places for all items in $I'$ of that size in the selected configurations. That is,
$$\sum_{c\in \cons} x_c\alpha_{cz} \geq n(z), \ \forall z\in\sizes$$

\paragraph*{The configuration linear program.}

The solution to the following configuration linear program gives a fractional selection of configurations.

\begin{align*}
	\min\quad        & \sum_{c\in\cons}x_c                                              \\
	\text{s.t.}\quad & \sum_{c\in \cons} x_c\alpha_{cz} \geq n(z), \ \forall z\in\sizes \\
	                 & x_c\geq 0,\ \forall c\in\cons
\end{align*}

\begin{theorem}\label{configlp cost}
Let $\opt'$ be the optimal solution to the rounded instance and also the cost of the solution.	Then, there is a feasible solution to the above linear program with cost at most $\opt'$.
\end{theorem}
\begin{proof}
	 We induce a solution to the linear program from $\opt'$. Do the following for each bin of $\opt'$. Identify the number of items of each distinct size in the bin and the number of idle items necessary. This gives the configuration corresponding to that bin. Let $\mathcal{C}$ be the multiset of all generated configurations from $\opt'$.

	For each configuration $c\in\cons$, identify the number of copies of $c$ in $\mathcal{C}$ and assign that to $x_c$. Since $\opt'$ is feasible, the $\textbf{x}$ values satisfy the constraint that there are enough places for all the items of each size.  Thus, $\textbf{x}$ is a feasible solution to the configuration linear program. Furthermore, $\sum_{c\in\cons}x_c$ is the total number of used bins, and it is $\opt'$. 
\qed\end{proof}

Next, we show how to approximate the above linear program in polynomial time. We use the column-generation technique of \cite{karmarkar1982efficient}.

\subsubsection{Approximating the configuration linear program.}

 The \emph{dual} of the configuration linear program is as below.
\begin{align*}
	\max\quad        & \sum_{z\in\sizes}n(z)y_z                                      \\
	\text{s.t.}\quad & \sum_{z\in \sizes} \alpha_{cz}y_z \leq 1, \ \forall c\in\cons \\
	& y_z\geq 0,\ \forall z\in\sizes
\end{align*}
In this dual linear program we have a dual variable $y_z$ for each $z\in\sizes$. We will approximate the dual using the ellipsoid algorithm. For that we require a polynomial time approximate separation oracle. The oracle should run in polynomial time and when given a candidate solution vector to the dual program $\textbf{y}^*$, the oracle should output either a constraint that is violated, or say that the solution is approximately feasible. A constraint corresponding to some configuration $c$ is approximately feasible if $\sum_{z\in \sizes} \alpha_{cz}y^*_z \leq \oeps$. If all constraints are approximately feasible, then $y^*/(\oeps)$ is a feasible solution to the dual program and we say that $y^*$ is {\em approximately feasible}. In order to output a constraint that is violated it is sufficient to output a configuration $c$ whose components are such that $\sum_{z\in \sizes} \alpha_{cz}y^*_z > 1$.

The approximate separation oracle is obtained as an approximation to an integer program. The $\alpha$ values will be the decision variables, the objective of the program will be to maximize $\sum_{z\in \sizes} \alpha_z y^*_z$. Let $\alpha'=\sum_{z\in\sizes}\alpha_z$ be a guessed cardinality bound of the configuration. The constraints will be derived from the properties of the configurations. The first constraint forces the number of items in the configuration to be equal to $\alpha'$, that is $\sum_{z\in\sizes}\alpha_{z}=\alpha'$. The second constraint is that the total size of the items in the configuration should be at most the bin size available after adding the idle items. That is, $\sum_{z\in\sizes}z\alpha_z\leq 1 - \id\left(\left\lceil\alpha'/k\right\rceil-1\right)$. The right-hand side is the space available after adding idle items corresponding to the cardinality of the item set. The integer program is as follows and is denoted as $IP(\alpha')$. 
\begin{align*}\tag{$IP(\alpha')$}
	\max\quad & \sum_{z\in \sizes} \alpha_z y^*_z                       \\
	\text{s.t.}\quad
	& \sum_{z\in\sizes}z\alpha_z \leq 1 - \id\left(\left\lceil\frac{\alpha'}{k}\right\rceil-1\right) \\
	& \sum_{z\in\sizes}\alpha_{z} = \alpha'                  \\
	& \alpha_z\in\nint,\ \forall z\in\sizes
\end{align*}
The program is approximated for all possible values of $\alpha'$, and since $\alpha'$ is an integer and $\alpha'\leq n$ there are at most a polynomial number of integer programs to consider.  If the objective value is strictly greater than $1$ (for at least one of the values of $\alpha'$), the configuration represented by the $\alpha$ values is a configuration that violates the constraint of the dual program.

$IP(\alpha')$ is a cardinality constrained knapsack problem and the problem can be solved approximately in polynomial time using the FPTAS in \cite{caprara2000approximation}. The problem statement for the FPTAS is as follows. There are $|\sizes|$ different item sizes, the value of an item of size $z\in \sizes$ is $y^*_z$, there are $n(z)$ items of size $z$, and we need to pack a subset of items such that the total number of items is $\alpha'$ and the size of the knapsack is $1 - \id\left(\left\lceil\alpha'/k\right\rceil-1\right)$. Then, we solve at most $n$ such problems (corresponding to different values of $\alpha'$). If there exists a solution with objective value strictly greater than $1$, then that solution corresponds to a configuration that violates the original dual constraint.  If none of these applications of the FPTAS for the cardinality constrained knapsack problem outputs a violated constraint of the dual, we conclude that the current dual solution vector is approximately feasible. 

Once we have an approximate solution to the dual we can get the solution to the primal problem with the same objective value in polynomial time. The ellipsoid algorithm for computing the dual solution has used a polynomial number of dual constraints (the violated constraints of all iterations). Thus, we can get a primal program with a polynomial number of variables by finding those constraints in the dual that were used by the ellipsoid algorithm. This results in a primal problem with polynomial number of constraints and variables. The modified primal program can be solved in polynomial time. Furthermore, we can output an optimal basic solution to this modified primal problem (with a polynomial number of variables) in polynomial time \cite{beling1998using}.  This optimal solution for the modified primal problem is our approximate solution to the (original) primal linear program.

Let the approximate solution to the primal linear program be $\textbf{x}^*$ and the objective value of this solution is at most $(\oeps)\opt'$. Next, we use this approximate solution to create a partition of $\items'$ (the items in $I'$) to bins. This is done by rounding up every component of the approximate solution to allow us to get an integral number of bins for each configuration. The algorithm to pack the items in $\items'$ integrally is as follows.

\subsubsection{Computing the output of the scheme.}
First, we round up the solution $\textbf{x}^*$ to $\textbf{x}'$ as follows. For each $c\in\cons$ in the support of $\textbf{x}^*$, we let $x'_c = \lceil x^*_c\rceil$. For each $c$ in the support of $\textbf{x}^*$, do the following. Initialize $x'_c$ empty bins and each of which is assigned the configuration $c$. The number of items of size $z\in\sizes$ in each bin will be at most the number of items of size $z$ in the corresponding configuration of that bin. For each distinct size $z\in\sizes$ and for each bin assigned configuration $c$, greedily pack $\alpha_{cz}$ items of size $z$ in $\items'$ to this bin unless there are not sufficiently unpacked  items of this size and in the last case we pack all remaining items of this size to the current bin. Observe that the load of each bin is at most $1$.

\begin{lemma}
	The above algorithm gives a feasible solution to \bpu{} for the rounded instance and the number of bins is at most $\sum_{c\in\cons}x'_c$.
\end{lemma}
\begin{proof}
	The rounding up of the solution of the linear program is feasible since $\sum_{c\in \cons} x'_c\alpha_{cz} \geq \sum_{c\in \cons} x^*_c\alpha_{cz}$ and since $\textbf{x}^*$ is a feasible solution. The number of places available for items of each distinct size is increased when we replace $x^*$ by $x'$.

	For a size $z\in\sizes$, according to the algorithm, the number of items of size $z$ placed in a bin is at most $\alpha_{cz}$, where $c$ is the configuration corresponding to the bin. The total number of items that can be placed in all bins assigned $c$ is at most $x'_c\alpha_{cz}$. Thus, if we assume by contradiction that for a size $z$ there are unassigned items (at the end of the procedure), then the total number of items of size $z$ placed in all bins is exactly $\sum_{c\in \cons} x'_c\alpha_{cz}$ and since $\sum_{c\in \cons} x'_c\alpha_{cz} \geq \sum_{c\in \cons} x^*_c\alpha_{cz} \geq n(z)$ we get a contradiction to the assumption. Therefore, all items are packed by the algorithm. The number of idle items that we need to add to a bin with configuration $c$ is at most $\delta_c$ and since $c$ was a feasible configuration, the load of each bin is at most $1$. 
\qed\end{proof}

Next, we upper bound $\sum_{c\in\cons}(x'_c - x^*_c)$ so that we can bound from above the cost of the solution returned by our scheme. 
\begin{lemma}
	The number of bins in the output of the scheme is at most $(1+2\eps)\opt + \ieps^3$.
\end{lemma}
\begin{proof}
A basic solution for the primal linear program has at most $\ieps^3-1$ (strictly) positive components since there are at most $\ieps^3-1$ constraints excluding the non-negativity constraints. Therefore, $\sum_{c\in\cons}(x'_c - x^*_c) \leq \ieps^3-1$, and by Lemma \ref{configlp cost}, we conclude that 
\begin{align*}\label{rem_items}
	\sum_{c\in\cons}x'_c = \sum_{c\in\cons} x^*_c + \sum_{c\in\cons} (x'_c-x^*_c)\leq  (\oeps)\opt + \ieps^3-1.
\end{align*}
From Lemma \ref{1_items}, $\items_1$ can be partitioned into at most $\eps\opt+1$ feasible bins. Thus, the number of bins in the output of the algorithm is at most
\begin{align*}
	(\oeps)\opt + \ieps^3-1 + \eps\opt +1 = (1+2\eps)\opt + \ieps^3.
\end{align*}
\qed\end{proof}

Informally, we conclude that if $\cb\leq \ieps^2$, then problem \bpu\ has an AFPTAS.

\subsection{Case II: $\cb > \ieps^2$}

In this case, our proof (from case I) of the upper bound on the increase of the cost due to the linear grouping of all items will not hold.  The reason is that here there are cases where the number of items in a bin can be non-constant. Thus, we will classify the items into small items and large items and perform linear grouping only on the large items. Such partition into small and large exists in earlier schemes for bin packing and for cardinality constrained bin packing. This will also mean that we will need to revise the definition of a configuration as well as the configuration linear program.  On the other hand we will use the fact that in this second case for sufficiently small items, it is possible to pack a set of $\frac 1{\eps}$ items in a common bin (while this fact does not hold in case I). Let $\opt$ be the optimal solution of the original instance (a partition of items in $\items$ into bins), and abusing the notation let it also denote the cost of the solution (the number of partitions or bins).

\subsubsection{Item classification and linear grouping of large items.}
We classify the items into small and large items. An item is \emph{small} if its size is less than $\eps$, and otherwise it is \emph{large}. The number of large items that can be packed in a bin is at most $1/\eps$. Let the set of large items be $\lar$ and the set of small items be $\sma$.

Next, we perform linear grouping as before but only for large items. First, sort $\lar$ in a non-increasing order of size (breaking ties arbitrarily). The items in $\lar$ are partitioned into sub-lists of consecutive (along the sorted list) items called {\em classes}. We have $\ieps^3$ classes denoted by $\lar_i, i\in[\ieps^3]$.   The items in $\lar$ are partitioned satisfying the following criteria: $\left\lceil \eps^3|\lar| \right\rceil = |\lar_1| \geq |\lar_2| \geq \ldots \geq |\lar_{\ieps^3}| = \left\lfloor \eps^3|\lar| \right\rfloor$.  Observe that we may have empty classes though (even if the number of items in the instance is non-constant). Then, we remove the items in class $\lar_1$ and pack each item of $\lar_1$ into a dedicated bin. Let $\lar'=\lar\backslash \lar_1$ be the set of the remaining large items. Next, we round the size of items in $\lar'$. The size of an item is rounded up to the size of the largest item in the class to which it belongs to. After the rounding, the number of distinct item sizes in $\lar'$ is at most $\ieps^3-1$. The rounded instance $I'$ is defined with the set of items $\lar'\cup \sma$ and the size of an item is its rounded size if it is a large item and the size of the item if it is a small item.

\begin{lemma}\label{c2_round}
	The cost of the optimal solution of the rounded instance is at most the cost of the optimal solution of the original instance, $\opt$.
\end{lemma}
\begin{proof}
	We define a new solution $\opt'$ for $I'$ based on $\opt$ without increasing the number of bins. The assignment of small items in $\opt'$ is the same as that of $\opt$. For each large item in $I'$, identify the class (of the linear grouping) to which it belongs. For $i\in[\ieps^3-1]$ and every $j$, we do the following. We place the rounded item that is the $j$-th item in the sublist consisting of the items of class $\lar_{i+1}$ in $\opt'$ in the place occupied by the $j$-th item in the sublist consisting of the items of class $\lar_i$ in $\opt$. The linear grouping satisfies the following two invariants.
	\begin{outline}[enumerate]
		\1 The number of items in each successive class is at most the preceding class, i.e. $|L_i| \geq |L_{i+1}|$.
		\1 The size of the items in each successive class is at most the size of the last item in the preceding class.
	\end{outline}
	Let $p'$ be a bin in $\opt'$ and let $p\in\opt$ be the bin with the same index as $p'$. The above two facts imply that the following properties hold.  First, every item is packed by the above rules. Second, the number of items in $p'$ will be at most the number of items in $p$. Lastly, the total rounded size of the items in $p'$ will be at most the total size of all items in $p$. Thus, the load of $p'$ will be at most the load of $p$, and since $p$ is a feasible bin, so does $p'$. Therefore, the number of (non-empty) bins in $\opt'$ is at most $\opt$.
\qed\end{proof}

Next, we upper bound the increase of the cost due to packing each item of $\lar_1$ into a dedicated bin.
\begin{lemma}\label{class_1}
	The packing of $\lar_1$ is into at most $\eps^2 \opt + 1$ bins.
\end{lemma}
\begin{proof}
The number of items in $\lar$ is
	\begin{align*}
		|\lar| \leq \frac{\sum_{i\in\lar}s_i}{\eps} \leq \opt/\eps,
	\end{align*}
where the first inequality holds because the size of each large item is at least $\eps$. Furthermore,
	\begin{align*}
		|\lar_1| \leq \eps^3|\lar| + 1 \leq \eps^2 \opt + 1,
	\end{align*}
where the first inequality is due to  the linear grouping. 
\qed\end{proof}

In what follows, we consider the rounded instance $I'$ and exhibit an AFPTAS for that.  By the last two lemmas, this suffices to obtain an AFPTAS for this case of \bpu\ as well. 

\subsubsection{Reformulation of \bpu.}
Similarly to our EPTAS for \jai\ we reformulate \bpu\ to help deal with the small items. The reformulated problem of \bpu\ is denoted as \jaii{}. The input instance to \jaii{} is the rounded instance $I'$. Consider a bin $p$, and sort its items in a non-increasing order of size breaking ties in an increasing order of indexes. An item in $p$ is \emph{early} if the item is among the first $k/\eps$ items in $p$, and an item is \emph{late} if the item is not among the first $k/\eps$ items in $p$. Let $\early_p$ and $\late_p$ represent the set of early items and late items of $p$, respectively. Since the total number of large items in a bin is at most $\ieps$ and the items are sorted in a non-increasing order of size, all the large items will be early items and only small items will be late items (if they exist). The unavailability requirement is modified so that it needs to be satisfied by only the early items.  This is not followed for the last at most $k$ early items if there are no late items. This is similar to the reformulation we did for \jai\ and thus we use the same notation of \jaii\ for the reformulated problem. Thus the number of idle items in a feasible bin in a solution to \jaii{} is at most $\ieps$. The size of an item $i$ is $s_i$ if $i$ is early, and otherwise the size of the item is $\ms_i = s_i+\id/k$. Thus, the load of a bin (or a set of items) $p$ is defined as
\[
L_{p} =
\begin{dcases}
	\sum_{i\in\early_{p}}s_i + \sum_{i\in\late_{p}}\ms_i + \dfrac{|\early_{p}|}{k}\id	& \text{ if } \late_{p}\not=\emptyset\\
	\sum_{i\in\early_{p}}s_i+\left\lfloor\dfrac{|\early_{p}|-1}{k}\right\rfloor\id & \text{ if }\late_{p}=\emptyset
\end{dcases}
\]
where $\early_{p}$ is the set of early items of $p$, and $\late_{p}$ is the set of late items of $p$. A \emph{feasible solution} to \jaii{} is a partition of items in $ I' $ into bins such that the load of each bin is at most $1$. Each bin $p$ has a partition into early items and late items of $p$. We denote by $\early_p$ and $\late_p$ the set of early items and the set of late items of $p$, respectively. The objective of \jaii{} is to minimize the number of bins in a feasible solution. The main difference between \bpu{} and \jaii{}, other than the definition of the load, is that we modify the size of an item if it is a late item. 

\begin{lemma}\label{atmost_eps}
	If there exists a late item in a bin of a feasible solution to \jaii, then there exists at least one set of $k$ consecutive early items in the bin such that the total size of these $k$ early items together with the size of the following idle item is at most $\eps$.
\end{lemma}
Before proving the lemma observe that the set of the smallest $k$ early items in the bin satisfies the condition.  Thus, the choice of such $k$ early items can be done in polynomial time.  Next, we turn our attention to proving the lemma. 
\begin{proof}
	We use a proof by contradiction. Consider a bin $p$ in a feasible solution to \jaii{} with at least one late item. Assume by contradiction that the total size of every set of $k$ consecutive early items together with one idle item is strictly greater than $\eps$. Then, the total size of all the $k/\eps$ early items together with $1/\eps$ idle items is strictly greater than $\ieps\cdot\eps=1$. Thus, there can be no late item contradicting our assumption.
\qed\end{proof}

An immediate consequence is that the size of an idle item $\id$ is at most $\eps$ if there is at least one late item in one of the bins of a feasible solution. Next, we are going to show that the reformulation can increase the total size of items in a bin by at most $\id$ because of at most $k$ late items (of modified size) and using the last lemma the cost of an optimal solution  increases by at most $\eps\opt+1$. 

\begin{lemma}\label{c2_reform}
	The cost of an optimal solution to the resulting instance of \jaii\ is at most $(\oeps)\opt + 1$.
\end{lemma}
\begin{proof}
	Let $\opt'$ be the optimal solution to \bpu{} for instance $I'$. The same partition is a candidate solution to \jaii{} but this candidate solution might be infeasible. We start with this partition into bins, and then transform it into a feasible solution. Let $\opt''$ be the solution to \jaii{} that we construct in the following process. We initialize the bins in $\opt''$ to be the item sets of the bins in $\opt'$. Let $B$ be the set of bins of $\opt''$ which have at least one late item. The modified size of a late item in $B$ is its size plus $\id/k$. $ \opt'' $ may become infeasible because the load of some bins may become greater than $1$. Let $p$ be a bin in $ \opt $, and let $\early_p$ and $\late_p$ be the set of early items and late items in the corresponding item set represented by $p$ in $\opt''$, respectively. If $p$ is not in $B$, then the load of $p$ in \bpu\ is the same as it is in \jaii, and since $p$ was feasible it remains feasible.  Consider $p\in B$. The increase in load of $p$ when we move from $\bpu$ to $\jaii$ is as follows.
	\begin{align*}
		&\sum_{i\in\early_{p}}s_i + \sum_{i\in\late_{p}}\ms_i + \frac{|\early_{p}|}{k}\id   - \left(\sum_{i\in p}s_i + \left\lfloor \frac{|p|-1}{k}\right\rfloor\id\right)\\
		&=\sum_{i\in\late_{p}}\frac{\id}{k} + \frac{|\early_{p}|}{k}\id - \left\lfloor \frac{|\early_{p}+\late_{p}|-1}{k}\right\rfloor\id\\
		&=\left(\frac{|\late_{p}|}{k} - \left\lfloor \frac{|\late_{p}|-1}{k}\right\rfloor\right)\id \leq \id.
	\end{align*}
	From Lemma \ref{atmost_eps}, this increase is at most $\eps$. Furthermore, from Lemma \ref{atmost_eps}, there exists a subset of $k$ early items $E_p$ in every bin $p\in B$ satisfying the claim of the lemma. We partition $B$ into at most $\eps|B|+1$ sublists of bins where every sublist contains at most $\ieps$ bins. For every sublist of bins, open a new bin $b$ and do the following. For every bin $p$ of the sublist, we move the items in $E_p$ together with an idle item (their total size will be at most $\eps$ and at least $\id$) to $b$. This new bin $b$ is feasible since the total size of items in the new bin is at most $\ieps\cdot\eps=1$ and the unavailability requirement is satisfied in this new bin. This transformation decreases the load of each bin that was in $B$ by at least $\id$ to be at most $1$. The bins in $ \opt'' $ along with the new bins are a feasible solution to \jaii. The cost of the solution is at most $\opt' + \left\lceil \opt'/(\ieps)\right\rceil \leq (\oeps)\opt +1$.
\qed\end{proof}

\subsubsection{The configuration linear program.}
We use a configuration linear program to fractionally assign bins to configurations and fractionally assign small items to configurations. Let the set of distinct item sizes of large items be $\sizes$ and $|\sizes|\leq\ieps^3-1$ from the linear grouping.

\paragraph*{Configurations.} A configuration encodes the information of a set of items packed into a bin in a feasible solution to \jaii{}. The set of large items packed into a bin is well-defined based on its configuration (up to assignment of equal sized items) but the set of small items is characterized using  approximate information. A configuration is a vector of length $|\sizes|+3$. The components of a configuration $c$ are as follows:
\begin{outline}[enumerate]
	\1 The first $|\sizes|$ components store the number of large items of a particular size packed in one bin represented by $c$. For each $z\in\sizes$, we have a component $\alpha_{cz}$ that stores the number of large items of size $z$ in such a bin.

	\1 The next component, denoted by $\delta_c$, stores the total number of early items rounded up to an integer multiple of $k$. That is, the total number of early items assigned to one bin with configuration $c$ is in the interval $((\delta_c-1)k, \delta_ck]$.


	\1 The next component denoted by $\gamma_{c}$ stores the total modified size of late items in a bin with configuration $c$ rounded down to an integer multiple of $\eps$. That is, the total modified size of late items in such a bin is in the interval $[\gamma_{c}\eps, (\gamma_{c}+1)\eps)$.  Note that the number of values possible for $\gamma_{c}$ for any configuration $c$ is at most $\ieps+1$.

	\1 The last component denoted by $\gamma'_c$ is  binary,  and it is $1$ if there is at least one late item, and otherwise it is $0$.  Observe that since we allow zero-sized items, we may have $\gamma'_c=1$ while $\gamma_c=0$.
\end{outline}


The load of a configuration $c$ is denoted as $L_c$ and is defined as
\begin{align*}
	L_c = \sum_{z\in\sizes}\alpha_{cz}z + \gamma'_c\gamma_c\eps + (\delta_c -1 + \gamma'_c)\id.
\end{align*}
A configuration $c$ is \emph{feasible} if $L_c\leq 1$. Let $\cons$ be the set of feasible configurations. Since $|\cons|$ is an exponential function of $\ieps$, we are not listing all configurations in $\cons$, and this is just notation to be used below.

Recall that a small item can be either an early item or a late item. In order to upper bound the total size of such small early items we use the components of the configuration along with the feasibility of a configuration. For $c\in\cons$, the total size of small items that are early items in a bin with configuration $c$ is at most
$$ 1-\left(\sum_{z\in\sizes}\alpha_{cz}z + \gamma'_c\gamma_c\eps + (\delta_c -1 + \gamma'_c)\id\right).$$


\paragraph*{The decision variables of the linear program.} The decision variables are $\textbf{u}$,$\textbf{v}$, and $\textbf{w}$. Their values mean the following.
\begin{outline}[enumerate]
	\1 $u_{c}$ is the number of times configuration $c$ is used.
	\1 $v_{i\eta}$ is the fraction of small item $i$ assigned as an early item to configurations $c$ whose $\gamma_c$ value satisfies $\gamma_c = \eta$.
	\1 $w_{i\eta}$ is the fraction of small item $i$ assigned as a late item to configurations $c$ whose $\gamma_c$ value satisfies $\gamma_c = \eta$.
\end{outline}
The number of $\textbf{v}$ and $\textbf{w}$ variables is at most $|\sma|\cdot\left(\ieps+1\right)$ each, since $\eta$ has at most $\ieps+1$ distinct values.  Note that here and unlike the EPTAS for \jai\ we use an aggregated form of the assignment variables of small items into configurations.  Our goal in this aggregation is to allow the next linear program to have one constraint per small item plus a constant number of additional constraints (excluding non-negativity constraints).

\paragraph*{The configuration linear program.}   Let $n(z)$ be the number of large items of size $z$ in $I'$. 
The configuration linear program is the following linear program. 

\begin{align}
	\min\quad        & \sum_{c\in\cons}u_c                                                                                                                                                                                                           \\
	\text{s.t.}\quad
	& \sum_{c\in\cons}u_c\alpha_{cz}
	\geq n(z),\ \forall z\in\sizes \label{c2_all_large}\\
	& \sum_{\eta\in \{0\}\cup[\ieps]}(v_{i\eta} + w_{i\eta})
	\geq 1,\ \forall i\in\sma\label{c2_all_small}\\
	& \sum_{i\in\sma}v_{i\eta}s_i
	\leq\sum_{c\in\cons:\gamma_c=\eta}u_c\left(1-\left(\sum_{z\in\sizes}\alpha_{cz}z + \gamma'_c\gamma_c\eps + (\delta_c -1 + \gamma'_c)\id\right)\right),\ \forall \eta\in \{0\}\cup[\ieps]\label{c2_small_early} \\
	& \sum_{i\in\sma}w_{i\eta}\ms_i
	\leq \sum_{c\in\cons:\gamma_c=\eta}u_c\gamma'_c(\gamma_c+1)\eps,\ \forall \eta\in \{0\}\cup[\ieps]\label{c2_small_late}\\
	& \sum_{c\in\cons:\gamma_c=\eta}\left( u_c\sum_{z\in\sizes}\alpha_{cz} + \sum_{i\in\sma}v_{i\eta} \right)
	\leq \sum_{c\in\cons:\gamma_c=\eta}u_c\delta_c k, \ \forall \eta\in \{0\}\cup[\ieps]\label{c2_unavailability}\\
	& u_c, v_{i\eta}, w_{i\eta} \geq 0
\end{align}

An intuitive description of the constraints follow. Constraints \eqref{c2_all_large} and \eqref{c2_all_small} ensure that all the large items and small items will be assigned, respectively. For each $\eta\in \{0\}\cup[\ieps]$, Constraint \eqref{c2_small_early} forces the total size of early small items assigned to configurations with $\gamma_c=\eta$ to be at most the total size available in all such bins. Similarly, Constraint \eqref{c2_small_late} forces the total modified size of the late small items assigned to configurations with $\gamma_c=\eta$ to be at most the total size available in all bins with such configurations. For each $\eta\in \{0\}\cup[\ieps]$, Constraint \eqref{c2_unavailability} limits the number of small items that can be assigned to bins with configurations $c$ satisfying that $\gamma_c=\eta$.

\begin{theorem}\label{configurationlp cost}
	If there exists a feasible solution $\opt'$ to \jaii{} of cost at most $\opt'$, then there exists a feasible solution to the linear program of cost at most $\opt'$. 
\end{theorem}
\begin{proof}
We exhibit a feasible solution to the linear program based on the item sets of the bins in $ \opt' $.

	\paragraph*{Defining a configuration for each bin.} We create a configuration $c(p)$ from each bin $p$ of $\opt'$. Let $\early_p$ and $\late_p$ be the set of early and late items in $p$, respectively.  For each $z\in\sizes$, identify the number of large items of size $z$ in $p$ and denote it by $\alpha_{c(p)z}$. $|\early_p|$ is the number of early items in $p$. We round up this number of early items in $p$  to the next integer multiple of $k$, and set it as $\delta_{c(p)}$. That is, $\delta_{c(p)} = \left\lceil |\early_{p}|/k\right\rceil$. If there is at least one late item, set $\gamma'_{c(p)} =1$, and otherwise set $\gamma'_{c(p)}=0$. Next, calculate the total modified size of late items, round it down to an integer multiple of $\eps$, and denote it by $\gamma_{c(p)}$. That is, $\gamma_{c(p)} = \left\lfloor \left(\sum_{i\in\late_{p}}\ms_i\right)/\eps\right\rfloor$. All these components together give the configuration $c(p)$ that represents $p$. Denote the multiset of all these configurations (of all bins) by $\mathcal{C}$.

	\paragraph*{Checking feasibility of the defined configurations.} Next, we will show that these defined configurations are indeed feasible. Recall that the load of $p\in\opt'$ (as a bin in \jaii{}) is defined as
	$$L_{p} =
	\begin{dcases}
		\sum_{i\in\early_{p}}s_i + \sum_{i\in\late_{p}}\ms_i + \dfrac{|\early_{p}|}{k}\id	& \text{ if } \late_{p}\not=\emptyset\\
		\sum_{i\in\early_{p}}s_i+\left\lfloor\dfrac{|\early_{p}|-1}{k}\right\rfloor\id & \text{ if }\late_{p}=\emptyset
	\end{dcases} $$
	Consider the configuration $c(p)$ that represents  $p$. We consider two cases, first we assume that there are no late items in $p$, and then we will assume that there is at least one late item in $p$. The difference between the  load of $c(p)$ and the load of $p$ when $\late_p =\emptyset$ is
	\begin{align*}
		L_p - L_{c(p)}
		&= \left(\sum_{i\in\early_{p}}s_i+\left\lfloor\dfrac{|\early_{p}|-1}{k}\right\rfloor\id\right) - \left(\sum_{z\in\sizes}\alpha_{c(p)z}z + (\delta_{c(p)} -1)\id\right) \\
		&=\sum_{\early_p\backslash\lar}s_i \geq 0.
	\end{align*}
	The difference in load of $c(p)$ and $p$ when there is at least one late item ($ p = \early_p\cup\late_p $ and $|\early_p|=k/\eps$) is
	\begin{align*}
		L_p - L_{c(p)}
		&= \left(\sum_{i\in\early_{p}}s_i + \sum_{i\in\late_{p}}\ms_i + \dfrac{|\early_{p}|}{k}\id\right) - \left(\sum_{z\in\sizes}\alpha_{c(p)z}z + \gamma_{c(p)}\eps + \delta_{c(p)}\id\right)\\
		&\geq \left(\sum_{i\in\early_{p}\backslash\sma}s_i + \sum_{i\in\late_{p}}\ms_i + \dfrac{|\early_{p}|}{k}\id\right) - \left(\sum_{z\in\sizes}\alpha_{c(p)z}z + \sum_{i\in\late_p}\ms_i + \frac{|\early_p|}{k}\id\right) = 0.
	\end{align*}
	Since $ \opt' $ was feasible, all the configurations generated as above are feasible.

	\paragraph*{The solution for the configuration linear program.} For each $c\in\cons$, identify the number of copies of $c$ in $ \mathcal{C} $ and denote it by $u_c$. For each small item $i$, do the following. Identify the bin $p$ in which $i$ is in, and find the configuration $c(p)$ defined for it. Let $\eta$ be the $\gamma_{c(p)}$ value of this identified configuration. If $i$ is an early item, set $v_{i\eta}=1$, and otherwise set $w_{i\eta}=1$.

	\paragraph*{Checking the feasibility of the generated solution.} Next, we check the feasibility of the generated solution $\textbf{u}, \textbf{v}, \textbf{w}$. Constraints \eqref{c2_all_large} and \eqref{c2_all_small} are satisfied because the solution we consider is feasible which implies that the solution is a partition of $I'$ and each small item is either early or late (but not both). Consider a bin $p$ and let $c(p)$ be its corresponding configuration. From the definition of $\gamma_{c(p)} = \left\lfloor \left(\sum_{i\in\late_{p}}\ms_i\right)/\eps\right\rfloor$, we get $\sum_{i\in\late_{p}}\ms_i \leq (\gamma_{c(p)}+1)\eps$. Recall that $\sum_{i\in\sma}w_{i\eta}\ms_i$ is the total modified size of all late items (only small items are late) in all bins whose corresponding configurations have $\gamma_c =\eta$. Considering all  such bins, we get $\sum_{i\in\sma}w_{i\eta}\ms_i \leq \sum_{c\in\cons:\gamma_c=\eta}u_c\gamma'_c(\gamma_{c} +1)\eps$. 
	Similarly, the total size of small early items in a bin $p$ whose corresponding configuration is $c(p)$ is $$\sum_{i\in\early_p\backslash\lar}s_i \leq 1-\left(\sum_{z\in\sizes}z\alpha_{c(p)z}+\id\left(\delta_{c(p)}-1+\gamma'_{c(p)}\right)+\gamma'_{c(p)}\gamma_{c(p)}\eps\right)$$  where the inequality holds by the feasibility of the bin $p$. $\sum_{i\in\sma}v_{i\eta}s_i$ is the total size of all small early items in all bins whose corresponding configurations have $\gamma_c =\eta$. Considering all such bins, we get $$\sum_{i\in\sma}v_{i\eta}s_i \leq \sum_{c\in\cons:\gamma_c=\eta} u_c\left(1 - \left(\sum_{z\in\sizes}\alpha_{cz}z+\gamma'_c\gamma_c\eps+(\delta_c-1+\gamma'_c)\id\right)\right)$$ for each value of $\eta$. Thus, Constraints \eqref{c2_small_early} and \eqref{c2_small_late} are satisfied. From the definition of $\delta_{c(p)} = \lceil |\early_{p}|/k \rceil$, we get $|\early_p| \leq \delta_{c(p)}k$. Considering all configurations $c$ with $\gamma_{c}=\eta$ with $u_c$ copies per configuration, $\sum_{c\in\cons:\gamma_c=\eta}\left(u_{c}\sum_{z\in\sizes}\alpha_{cz} + \sum_{i\in\sma}v_{i\eta}\right) \leq \sum_{c\in\cons:\gamma_c=\eta}u_{c}\delta_{c}k$ for each value $\eta$. Thus, Constraint \eqref{c2_unavailability} is satisfied by the solution. The non-negativity constraints are satisfied by the definition of the decision variables. Thus, the generated solution is feasible.

	The cost of the solution is $\sum_{c\in\cons}u_c$ which is the number of bins in the solution for \jaii\ that we considered. Thus, the cost of the optimal solution for the configuration linear program is at most  $\opt'$.
\qed\end{proof}

Let the cost of the optimal solution and the optimal solution of the above linear program be denoted by $\opt_p$. An approximate solution $\opt'_p$ whose cost is also denoted by $\opt'_p$ such that $\opt'_p \leq (\oeps)\opt_p$ can be obtained in polynomial time. This is explained next. Furthermore, we assume without loss of generality that $\opt'_p$ is a basic solution for this linear program.  Similarly to case I, we use the column-generation technique of \cite{karmarkar1982efficient}.

\subsubsection{Approximating the configuration linear program.}
For ease of writing, for a configuration $c$ we denote
$1-\left(\sum_{z\in\sizes}\alpha_{cz}z + \gamma'_c\gamma_c\eps + (\delta_c -1 + \gamma_c')\id\right)$ as $\beta_c$. The primal problem is the configuration linear program above and using this notation of $\beta_c$ it is the following one.

\begin{align*}
	\min\quad        & \sum_{c\in\cons}u_c\\
	\text{s.t.}\quad
	& \sum_{c\in\cons}u_c\alpha_{cz} \geq n(z),\ \forall z\in\sizes\\
	& \sum_{\eta\in \{0\}\cup[\ieps]}\left(v_{i\eta}+w_{i\eta}\right) \geq 1,\ \forall i\in\sma\\
	& \sum_{c\in\cons:\gamma_c=\eta}u_c\beta_c-\sum_{i\in\sma}v_{i\eta}s_i \geq 0,\ \forall \eta\in \{0\}\cup[\ieps]\\
	& \sum_{c\in\cons:\gamma_c=\eta}u_c\gamma'_c(\gamma_c+1)\eps-\sum_{i\in\sma}w_{i\eta}\ms_i \geq 0,\ \forall \eta\in \{0\}\cup[\ieps]\\
	& \sum_{c\in\cons:\gamma_c=\eta}u_c\delta_c k-\sum_{c\in\cons:\gamma_c=\eta}\left( u_c\sum_{z\in\sizes}\alpha_{cz} + \sum_{i\in\sma}v_{i\eta} \right) \geq 0, \ \forall \eta\in \{0\}\cup[\ieps]\\
	& u_c, v_{i\eta}, w_{i\eta} \geq 0
\end{align*}
The dual linear program is as follows.
\begin{align}
	\max\quad
	&	\sum_{z\in\sizes}n(z)\lambda_z + \sum_{i\in\sma}\mu_i\\
	\text{s.t.}\quad
	&	\sum_{z\in\sizes}\alpha_{cz}\lambda_z + \nu_{\eta}\beta_c + \xi_{\eta}\gamma'_c(\gamma_c+1)\eps + \rho_{\eta}\left(\delta_ck-\sum_{z\in\sizes}\alpha_{cz}\right)\leq 1, \forall \eta, c: \gamma_c= \eta\label{dual_one}\\
	&	\mu_i-s_i\nu_{\eta} + \rho_{\eta}   \leq 0, \forall \eta, i\label{dual_two}\\
	&	\mu_i-\ms_i\xi_{\eta}  \leq 0, \forall \eta,i\label{dual_three}
\end{align}

The dual can be approximated in polynomial time using the ellipsoid method as follows. To do this we need a separation oracle that performs approximate feasibility checks in polynomial time when given a candidate solution vector $(\lambda^*, \mu^*,\nu^*, \xi^*, \rho^*)$. Constraint \eqref{dual_one} is approximately feasible for some configuration $c$ if $\sum_{z}\alpha_{cz}\lambda^*_z + \nu^*_{\eta}\beta_c + \xi^*_{\eta}\gamma'_c(\gamma_c+1)\eps + \rho^*_{\eta}\left(\delta_ck-\sum_{z\in\sizes}\alpha_{cz}\right)\leq 1+\eps$.  Other constraints are approximately feasible if they are feasible (note that there is a polynomial number of such additional constraints, and we can check that they are satisfied in polynomial time). If all constraints are approximately feasible, then $1/(\oeps)\cdot(\lambda^*, \mu^*, \nu^*, \xi^*, \rho^*)$ is a feasible solution to the dual program. Next, we explain how to obtain such a polynomial time approximate separation oracle.

\paragraph*{Approximate Separation Oracle.}

When given a dual candidate solution, constraints \eqref{dual_two} and \eqref{dual_three} can be checked in polynomial time whether they are violated or not. Thus, in order to apply the ellipsoid method, we need to provide a polynomial time approximate separation oracle for the constraint \eqref{dual_one}.

Consider one such constraint for a given dual solution $(\lambda^*, \mu^*,\nu^*, \xi^*, \rho^*)$:
\begin{align*}
	\sum_{z\in\sizes}\alpha_{cz}\lambda^*_z + \nu^*_{\eta}\left(1-\left(\sum_{z\in\sizes}\alpha_{cz}z + \gamma_c\eps + (\delta_c -1 + \gamma'_c)\id\right)\right) + \xi^*_{\eta}\gamma'_c(\gamma_c+1)\eps + \rho^*_{\eta}\left(\delta_ck-\sum_{z\in\sizes}\alpha_{cz}\right)\leq 1
\end{align*}
We would like to apply some guessing of partial information on the configuration in order to allow us to approximate a constant number of simple integer programs.  To do that, we rearrange the terms, and get the following equivalent constraint.
\begin{align*}
	\sum_{z\in\sizes}\alpha_{cz}(\lambda^*_z-\nu^*_{\eta}z)\leq 1 - \left(-\rho^*_{\eta}\sum_{z\in\sizes}\alpha_{cz} + \nu^*_{\eta}   + \left(\xi^*_{\eta}\gamma'_c(\gamma_c+1)\eps-\nu^*_{\eta} \gamma_c\eps\right) + \left(\rho^*_{\eta}\delta_ck-\nu^*_{\eta} (\delta_c -1 + \gamma_c')\id\right)\right)
\end{align*}
If we can find a configuration with $\alpha$, $\gamma$, $\gamma'$, and $\delta$ values such that for the given candidate solution vector $(\lambda^*, \mu^*, \nu^*, \xi^*, \rho^*)$, $$\sum_{z\in\sizes}\alpha_{z}(\lambda^*_z-\nu^*_{\eta}z) > 1 - \left(-\rho^*_{\eta}\sum_{z\in\sizes}\alpha_{z} + \nu^*_{\eta}   + \left(\xi^*_{\eta}\gamma'(\gamma+1)\eps-\nu^*_{\eta} \gamma\eps\right) + \left(\rho^*_{\eta}\delta k-\nu^*_{\eta} (\delta -1 + \gamma')\id\right)\right),$$ then the configuration defined by those $\alpha$, $\gamma$, $\gamma'$, and $\delta$ values violates the dual constraint.

Let $\alpha' = \sum_{z\in\sizes}\alpha_{z}$ be a guessed cardinality bound of the large items in a configuration. We have at most $\ieps+1$ values for $\alpha'$, $\gamma$ has at most $\ieps+1$ values, $\gamma'$ has at most $2$ values, and $\delta$ has at most $n$ values (since they are integers). We enumerate all possibilities of this information.  For a given value of this vector of information consisting of $(\alpha',\gamma,\gamma',\delta)$, we will check if either every valid configuration corresponding to this guess approximately satisfies the above constraint, or we find a violated constraint. Observe that such guessed information is meaningful only if $\alpha'\leq \delta k$.
We use the following integer program. 

\begin{align}
\max & \ \ \ \ \ \ \  \sum_{z\in\sizes}\alpha_{z}(\lambda^*_z-\nu^*_{\eta}z) \\
s.t. \ \ \ \	\sum_{z\in\sizes}\alpha_z &\leq \alpha'\label{cardinality check}\\
	\sum_{z\in\sizes}\alpha_z z &\leq 1-\gamma'\gamma\eps-\id(\delta-1+\gamma')\label{second_last}\\
	\alpha_z&\in\nint, \forall z\in\sizes
\end{align}
The decision variables are $\alpha_z,\forall z\in\sizes$. The constraints of the integer program have the following meaning. Constraint \eqref{cardinality check} checks if we can assign large items to match the guessed cardinality bound.  Constraint \eqref{second_last} checks if all the large items fit into the bin together with the small items and idle items based on the $\gamma$ and $\delta$ components. We next show how to approximate the program in polynomial time so that the total time (time taken to approximate a polynomial number of such integer programs) will also be polynomial in the encoding length of $I'$ as an instance for \bpu.

We round down the $z$ values to integer multiple of $\frac{\eps}{n}\cdot\left(1-\gamma\eps-\id(\delta-1+\gamma')\right)$ to use the methodology of \cite{jansen2018integer, eisenbrand2019proximity}. Then for each integer multiple (at most polynomial number) of $\eps/n$ of the right hand side of the Constraints \eqref{second_last}, we will get an integer program with constant number of constraints (two constraints in our case). Each such rounded program can be solved in polynomial time using \cite{jansen2018integer, eisenbrand2019proximity}.

More precisely, let $z'$ be the rounded values of $z$, that is, for all $z\in \sizes$, let
\begin{align*}
	z' 		&= 	\left\lfloor \frac{z}{\dfrac{\eps\cdot(1-\gamma\eps-\id(\delta-1+\gamma'))}{n}} \right\rfloor\frac{\eps\cdot(1-\gamma\eps-\id(\delta-1+\gamma'))}{n}.
\end{align*}
Then the rounded integer program is
\begin{align}
\max & \sum_{z\in\sizes}\alpha_{z}(\lambda^*_{z}-\nu^*_{\eta}z') \\
s.t. &\sum_{z\in\sizes}\alpha_{z}  \leq \alpha'\\
	&\sum_{z\in\sizes}\alpha_{z}z' \leq 1-\gamma'\gamma\eps-\id(\delta-1+\gamma')\label{new_second_last}\\
	&\alpha_{z}\in\nint, \forall z\in\sizes
\end{align}

\begin{lemma}
Consider a solution $\alpha^*$ that is feasible to the rounded integer program  whose objective function value as a solution for the rounded integer program is $A$. Then, $\alpha^*/(1+\eps)$  is an approximate solution to the original integer program whose objective function value as a solution for the original integer program is at least $A/(1+\eps)$.
\end{lemma}
\begin{proof}
First, consider the objective function value and using that $z'\leq z$ and $\nu^*_{\eta}\geq 0$, we conclude that the objective function value of $\alpha^*$ with respect to the original objective function is not smaller than it is with respect to the rounded objective function.  Thus the claim regarding the objective function value follows. Constraint \eqref{cardinality check} is satisfied as the same constraint is satisfied also by $\alpha^*$. Consider constraint \eqref{second_last}. From the definition of $z'$, $z' >  z - \eps/n\cdot(1-\gamma\eps-\id(\delta-1+\gamma'))$. The left hand side of the rounded constraint \eqref{new_second_last} is
	\begin{align*}
		1-\gamma'\gamma\eps-\id(\delta-1+\gamma')
		\geq \sum_{z\in\sizes}\alpha^*_{z} z'
		&\geq \sum_{z\in\sizes}\alpha^*_{z} \left(z - \frac{\eps}{n}(1-\gamma'\gamma\eps-\id(\delta-1+\gamma'))\right)\\
		&\geq \sum_{z\in\sizes}\alpha^*_{z} z - \eps(1-\gamma'\gamma\eps-\id(\delta-1+\gamma')),
	\end{align*}
	where the last inequality is because the total number of large items in a configuration is at most $n$ and we get
	\begin{align*}
		\sum_{z\in\sizes}\alpha^*_{z} z \leq  (1+\eps)(1-\gamma'\gamma\eps-\id(\delta-1+\gamma')).
	\end{align*}
	Thus $\alpha^*/(1+\eps)$ satisfies the original constraint \eqref{second_last}.
\qed\end{proof}
Furthermore, we have the following.
\begin{lemma}
Consider a solution $\alpha^*$ that is feasible to the original integer program  whose objective function value as a solution for the original integer program is $A$. Then, $\alpha^*$  is a feasible solution to the rounded integer program whose objective function value as a solution for the rounded integer program is at most $A$.
\end{lemma}
\begin{proof}
The feasibility of $\alpha^*$ in the rounded integer program follows by the fact that the left hand side of every constraint in the rounded integer program is not larger than it is in the original integer program and the right hand sides of the two programs are the same.  With respect to the objective function values, the claim follows by noting that the coefficient multiplying the $\alpha^*_z$ value in the rounded integer program is not smaller than it is in the original integer program.  Given the non-negativity of $\alpha^*$, the claim follows. 
\qed\end{proof}
Thus, in order to approximate the original integer program, it suffices to solve (exactly) the rounded integer program, and indeed we obtain an approximate separation oracle for the dual linear program, as we claimed.

As in the previous case, we can output a basic solution for the primal linear program that is an approximate solution for the configuration linear program. Let $(\textbf{u}^*, \textbf{v}^*, \textbf{w}^*)$ be this approximate solution to the configuration linear program. Thus, its objective value is at most $\opt'_p$ where $\opt'_p\leq (\oeps)\opt_p$.  

\subsubsection{Output of the scheme.}

We now have a fractional solution of the configuration linear program that gives some information regarding the packing of the items. Based on this we will generate an integral packing of the items to bins.

\paragraph*{Rounding up the solution.}

Let the fractional solution be $(\textbf{u}^*, \textbf{v}^*, \textbf{w}^*)$. Round up $\textbf{u}^*$ to $\textbf{u}'$ as follows. For all components in the support of $\textbf{u}^*$, set $u'_c = \lceil u^*_c\rceil$. Assign configuration $c$ to $u'_c$ bins for all $c$ in the support of $\textbf{u}'$. Define $\Delta = \left[\sum_{c\in\cons}u'_c\right]$ as the index set for these bins. 

\begin{lemma}\label{c2 rounding up cost}
	Rounding up the solution from $\textbf{u}^*$ to $\textbf{u}'$ increases the cost of the solution by at most $3\left(\ieps+1\right)+\ieps^3$.
\end{lemma}
\begin{proof}
The increase in the number of bins is $\sum_{c\in \cons} (u'_c-u^*_c)$.  The last term is at most the number of nonzero fractional components of the solution that is at most the number the number of constraints $3\left(\ieps+1\right)+\ieps^3+|\sma|$. But for the Constraint \eqref{c2_all_small}, at least one component for each small item is strictly positive. Thus the number of other strictly positive components in the basic solution is at most $3\left(\ieps+1\right)+\ieps^3$ and this upper bounds the increase in the cost of the solution as we claimed.
\qed\end{proof}

\paragraph*{Creating an integral packing of large items.}
First, we will assign the large items to these bins based on their configurations.
As a result of the rounding we now have more places for large items than necessary. For each bin $b\in\Delta$, do the following. Let $c$ be the configuration assigned to $b$. For each distinct size $z\in\sizes$, identify the number of unpacked items of size $z$, $n_z$ ($n_z$ can be $0$). If $n_z\geq \alpha_{cz}$ pack $\alpha_{cz}$ items of size $z$ to this bin, and otherwise assign $n_z$  items of size $z$ to this bin. From Constraint \eqref{large_places} and because $\textbf{u}'$ was rounded up, all the large items are packed integrally to the bins using this polynomial time procedure.

\paragraph*{An overview of the packing of small items.}
Our next objective is to pack the small items integrally. To that extent we first assign the small items fractionally to bins based on the fractional solution to the configuration linear program. The resulting fractional assignment will be converted to an integral packing while opening some new bins but the total number of bins after packing all the small items will be at most $(1+O(\eps))\opt$.

\paragraph*{Creating a fractional assignment of small items to bins.}
For every bin $b\in \Delta$, we let $c(b)$ be its assigned configuration.  
This configuration encodes an upper bound $\zeta^e_b$ on the total size of small items packed into a bin with configuration $c(b)$ as early items (see the discussion following the definition of configurations), as well as their cardinality and we denote this second upper bound by $\zeta^n_b$. But $c(b)$ also encodes the upper bound $\zeta_{b}^l$ on the total modified size of small items packed into such a bin as late items where this upper bound is $(\gamma_{c(b)}+1)\eps$.  Using these three upper bounds for each bin, we get that there is a feasible fractional solution for the following feasibility linear program that has variables  $v_{i,b},w_{i,b}$ for every small item $i$ and every bin $b$.

\begin{eqnarray*}
\sum_{i\in \sma} s_i v_{i,b} & \leq \zeta^e_b& \forall b\in \Delta \\
\sum_{i\in \sma} v_{i,b} & \leq \zeta^n_b& \forall b\in \Delta \\
\sum_{i\in \sma} \ms_i w_{i,b} & \leq \zeta^l_b& \forall b\in \Delta \\
\sum_{b\in \Delta} (v_{i,b}+w_{i,b}) &=1& \forall i\in \sma\\
v_{i,b},w_{i,b} &\geq 0& \forall i\in \sma, \ \forall b\in \Delta .   
\end{eqnarray*}
The feasibility of the last linear program follows trivially by the fact that we rounded up $u^*$ before assigning configurations to bins.  
Our next step is to find a basic feasible solution $(v^*,w^*)$ for the last set of constraints.  The number of strictly positive variables in this basic solution is at most $|\sma|+3|\Delta|$.  For every $i\in \sma$ we have that at least one of the variables of the form $v_{i,b}$ or of the form $w_{i,b}$ (for at least one bin $b$) is strictly positive.  However, if there is only one of those variables that is strictly positive then it must be equal to $1$ so it is integral.  Thus, the number of items for which there is a fractional decision variable, is at most $3\Delta$.  Our next step is to round down $v^*_{i,b},w^*_{i,b}$, and for every item $i$ for which there is a corresponding variable  either $v_{i,b}$ or $w_{i,b}$ that equals $1$ we pack the item to the corresponding bin $b$.  We are left with at most $3|\Delta|$ items but all of these items are small.

To create the integral packing of items to bins we partition the remaining small items into $6\eps|\Delta| +1$ new dedicated bins each of which has at most $\frac{1}{2\eps}$ small items that were not packed earlier.  Observe that since the items are small all these bins are feasible since we are in the second case of $\cb$.  This is so by the following argument.  If $k\leq \frac{1}{2\eps}$, since $\cb>\ieps^2$, $\id\leq \eps$ and every set of $1/2\eps$ small items together with at most $1/2\eps$ idle items fit into a common bin.  Otherwise, $k>\frac{1}{2\eps}$ and every subset of $1/2\eps$ small items fit into a common bin.

\paragraph*{Creating a feasible assignment of small items.}
The bins in $\Delta$ are packed by the previous step.  However, for a bin containing some late small items the total size of items together with the idle items may exceed the capacity of the bin.  This may occur as we have used a rounded value of $\gamma$.  Observe however that if there are no late items in the bin then the set of items packed there by the previous step indeed fits into one feasible bin (together with the corresponding idle items).

 For each bin $b\in\Delta$ that has late items we do the following. We move small items of total modified size of at least $2\eps$ (one additional $\eps$ to account for the increase of $\eps$ in the total size of small items) and at most $4\eps$ out of the bin $b$. This will reduce the total item size on each bin to be at most $1$. Notice that a small item has  a modified size that is its size increased by $\id/k$ and when $k$ such items are packed together we automatically obtain an idle item of length $\id$ (by rearranging the idle item fractions) satisfying the unavailability requirement. Such small items from $1/4\eps$ bins can be packed in a new bin, since it will satisfy the unavailability requirement and the total size is at most $1/4\eps\cdot4\eps=1$. This will require at most $4\eps\Delta + 1$ additional bins. 

\begin{lemma}\label{c2 final bins}
	The final number of bins used in the output of the scheme is at most $(1+O(\eps))\opt+f(\ieps)$.
\end{lemma}
\begin{proof}
	From lemma \ref{c2_reform}, \ref{c2 rounding up cost}, and theorem \ref{configurationlp cost}, $$\Delta \leq (\oeps)\left((\oeps)\opt+1\right)+ 3(\ieps+1)+\ieps^3 \leq (1+3\eps)\opt+5+3/\eps+\ieps^3.$$ The final number of bins used for items in $\items'$ is $6\eps\Delta+1+4\eps\Delta+1+\Delta \leq (1+10\eps)\Delta+2 $. Additional $\eps\opt+1$ bins are needed for packing the items in $\lar_1$, and we get that the final number of used bins is at most $(1+O(\eps))\opt+f(\ieps)$.
\qed\end{proof}

\begin{theorem}
	Problem \bpu{} admits an AFPTAS.
\end{theorem}
\begin{proof}
	All the operations were performed in time polynomial in the input encoding length. From lemma \ref{c2 final bins}, the number of bins in the output of the scheme is at most $(1+O(\eps))\opt+f(1/\eps)$.
\qed\end{proof}

\bibliographystyle{abbrv}

\end{document}